\documentclass[11pt,draftcls,onecolumn]{IEEEtran}

\author{Arash Mohammadi,  \textit{Student Member, IEEE}
and Amir Asif, \textit{Senior Member, IEEE} \\
Computer Science and Engineering, York University,
Toronto, ON, Canada M3J 1P3\\
Emails: $\{$marash, asif$\}$@cse.yorku.ca,
Ph: (416) 736-2100 Ext. 70128, Fax: (416) 736-5937}

\title{Distributed Particle Filter Implementation with Intermittent/Irregular Consensus Convergence}
\usepackage[dvips]{graphics}
\usepackage{epsfig}
\usepackage{amsmath}
\usepackage{amsthm}
\usepackage{amssymb }
\usepackage{amsmath}
\usepackage{bm}
\usepackage{float}
\usepackage{mathrsfs}
\usepackage{amsfonts}
\usepackage{graphicx}
\usepackage{algorithm, algpseudocode}
\usepackage{subfigure}
\def\x{\bm{x}}
\def\f{\bm{f}}
\def\g{\bm{g}}
\def\X{\bm{X}}

\def\z{\bm{z}}

\def\cx{{{\mathbb{X}}}}
\def\GF{\text{FF}}
\def\MGF{\text{MFF}}
\def\asym{\text{asym}}
\def\ggg{\text{FF}}

\newtheorem{thm}{Theorem}
\newtheorem{prop}{Proposition}
\algnewcommand\Input{\item[\hspace{6pt}\textbf{Input:}]}
\algnewcommand\Output{\item[\hspace{6pt}\textbf{Output:}]}
\algnewcommand\OutputVal{\textbf{output} }
\graphicspath{{}}

\begin{document}

\date{\today}
\maketitle
\thispagestyle{empty}

\begin{abstract}
Motivated by non-linear, non-Gaussian, distributed multi-sensor/agent navigation and tracking applications, we propose a multi-rate consensus/fusion based framework for distributed implementation of the particle filter (CF/DPF). The CF/DPF framework is based on running localized particle filters to estimate the overall state vector at each observation node. Separate fusion filters are designed to consistently assimilate the local filtering distributions into the global posterior by compensating for the common past information between neighbouring nodes.
The CF/DPF offers two distinct advantages over its counterparts. First, the CF/DPF framework is suitable for scenarios where network connectivity is intermittent and consensus can not be reached between two consecutive observations. Second, the CF/DPF is not limited to the Gaussian approximation for the global posterior density.
A third contribution of the paper is the derivation of
the exact expression for computing the
posterior Cram\'er-Rao lower bound~(PCRLB) for the distributed architecture based on a recursive procedure
involving the local Fisher information matrices (FIM) of the distributed estimators.
The performance of the CF/DPF algorithm~closely follows the centralized particle filter approaching the PCRLB at the signal to noise ratios that we tested.
\end{abstract}
\textbf{IEEEkeywords:} Consensus algorithms, Data fusion, Distributed estimation, Multi-sensor tracking, Non-linear systems, and Particle filters.
\ifCLASSOPTIONpeerreview
\begin{center} \bfseries EDICS Category: SSP-FILT; SSP-NGAU;SSP-PERF; SEN-FUSE; SEN-DIST; SSP-NPAR; SSP-TRAC \end{center}
\fi
\section{Introduction} \label{sec:Introduction}

\IEEEPARstart{T}{he} paper focuses on distributed estimation and tracking algorithms
for non-linear, non-Gaussian, data fusion problems in networked systems. Distributed state estimation has been the~center of
attention recently both for
linear~\cite{Olfati:3}--\nocite{Kar:2011,Alriksson:2006,Olfati:track,Olfati:4}\cite{Dimakis:2010}
and non-linear
systems~\cite{Challa:2002}--\nocite{Rosencrantz:2003,Rigatos:2010,  coat:2003,sheng:2005,dpd:durant,CDPF:Durant,DongbingGu:2007,consdpf:1,Coates:2010,Ustebay:2011,Bolic:2005,Balasingam:2011,Arash:2009,Simonetto:2010,Huang:2008,Arash:ssp2011, Arash:campsap2011, Farahmand:2010, Arash:Icc2012, Lee:2009, Hlinka:2010}\cite{Liu:2009}
with widespread applications such as autonomous navigation of unmanned aerial vehicles (UAV)~\cite{Rigatos:2010}, localization in robotics~\cite{Simonetto:2010}, tracking/localization in underwater sensor networks~\cite{Huang:2008}, distributed state estimation  for power distribution networks~\cite{Arash:2009}, and bearings-only target tracking~\cite{Bolic:2005}.
A major problem in distributed estimation networks is unreliable communication (especially in large and
multi-hop networks), which results in communication delays and
 information loss. Referred to as intermittent network connectivity~\cite{Sinopoli:2004,  Kar:2011int}, this issue has been investigated broadly in the context of the Kalman filter~\cite{Sinopoli:2004,  Kar:2011int}. Such methods are, however,~limited to linear systems and have not yet  been extended to non-linear systems. The paper addresses this gap.

\noindent
\textbf{Distributed Estimation:} Traditionally, state estimation algorithms have been largely centralized with
participating nodes communicating their raw observations (either
directly or indirectly via a multi-hop relay) to the fusion centre: a central processing
unit responsible for computing the overall estimate. Although optimal,
such centralized approaches are unscalable and susceptible to failure in case the
fusion centre breaks down.
The alternative is to apply  distributed estimation algorithms, where: (i)~There is no fusion center;
(ii)~The sensor nodes do not require global knowledge of the network
topology, and; (iii)~Each node exchanges data only within its immediate neighbourhood limited to a few local nodes.
The distributed estimation approaches fall under two
main categories:~Message passing
schemes~\cite{coat:2003,sheng:2005}, where
information flows in a sequential, \textit{pre-defined} manner from a node to
one of its neighboring nodes via a cyclic path till the entire network
is traversed, and;~Diffusive schemes~\cite{Arash:2009}--\nocite{Simonetto:2010,Huang:2008}\cite{Bolic:2005},~\cite{dpd:durant}--\nocite{CDPF:Durant,DongbingGu:2007,consdpf:1,Coates:2010,Ustebay:2011,Balasingam:2011,Farahmand:2010}\cite{Arash:Icc2012}, where each node communicates its
local information  by interacting with its immediate
neighbors.  In dynamical environments with intermittent connectivity, where frequent changes in the
underlying network topology due to mobility, node failure, and link
failure are a common practice, diffusive approaches
significantly improve the robustness at the cost of an increased
communications overhead.
In diffusive schemes, the type of information communicated across the network varies from local observations, their likelihoods, or some other function of local observations~\cite{coat:2003, DongbingGu:2007, Ustebay:2011, Farahmand:2010, Arash:Icc2012, Hlinka:2010}, to state posterior/filtering estimates evaluated at individual nodes~\cite{sheng:2005, dpd:durant, consdpf:1,Coates:2010, Arash:ssp2011, Arash:campsap2011}.
Communicating state posteriors  is advantageous over sharing likelihoods in applications with intermittent connectivity. In theory, any loss of  information in error prone networks
is contained in the following posteriors and is, therefore, automatically compensated for as the distributed algorithms iterate.
A drawback of communicating the local state estimates stems from their correlated nature~\cite{Grime:1994}.
Channel filters~\cite{Grime:1994} and their non-linear extensions~\cite{dpd:durant} (proposed to ensure consistency of the fused estimates by removing this redundant past information) associate an additional filter to each communication link and track the redundant information between a pair of  neighbouring nodes' local estimates. However, channel filters are limited to tree-connected topologies and can not be generalized to random/mobile networks.
Alternatively, consensus-based approaches have been recently introduced to extend distributed estimation to arbitrary network topologies
with the added advantage that the algorithm is somewhat immune to node and/or communication failures~\cite{Olfati:4, Dimakis:2010}.
The consensus-based\footnote{Consensus in distributed filtering is the process of establishing a consistent value for some statistics of the state vector across the network by interchanging relevant information between the connected neighboring nodes.} distributed Kalman filter~\cite{Olfati:3}--\nocite{Kar:2011, Alriksson:2006, Olfati:track,Olfati:4}\cite{Dimakis:2010} have been widely explored for estimation/tracking
problems in linear systems with intermittent connectivity but there is still a need for developing distributed estimation
approaches for non-linear systems.
In addition, the current non-linear consensus-based distributed approaches suffer from three drawbacks.
First, the common practice~\cite{DongbingGu:2007}--\nocite{consdpf:1}\cite{Coates:2010}\cite{Arash:ssp2011} of limiting the global posterior to Gaussian distribution is sub-optimal.
Second, common past information between neighbouring nodes gets incorporated multiple times~\cite{consdpf:1}.
Finally,  these approaches~\cite{DongbingGu:2007}--\nocite{consdpf:1, Coates:2010, DongbingGu:2007, Ustebay:2011,Farahmand:2010}\cite{Arash:Icc2012}, require the consensus algorithm to converge between two successive observations (thus  ignoring the intermittent communication connectivity issue in the observation framework).
The performance of the distributed approaches degrade substantially if consensus is not reached within two consecutive observations.

Motivated by distributed navigation and tracking applications
within large networked systems,
we propose
a multi-rate framework for distributed implementation of the particle filter.
The proposed framework is suitable for scenarios where the network connectivity is intermittent and consensus can not be reached between two observations. Below, we summarize the key contributions of the paper.

\noindent
\textbf{1. Fusion filter:}
The paper proposes a consensus/fusion based distributed implementation
of the particle filter (CF/DPF) for non-linear systems with
non-Gaussian excitation. In addition to the localized particle filters, referred to as the local filters, the CF/DPF introduces separate
consensus-based filters, referred to as the fusion filters (one per sensor node), to derive the global
posterior distribution by consistently fusing local filtering densities in a distributed fashion.
The localized implementation of the particle filter and the fusion filter used to achieve consensus are run in parallel, possibly at  different rates. Achieving consensus between two successive iterations of the local filters is, therefore, no longer a requirement.
The CF/DPF compensates for the common past information between local
estimates based on an optimal non-linear Bayesian fusion rule~\cite{Chong:1990}.
The fusion concept used in the CF/DPF is similar to~\cite{Grime:1994}
and~\cite{dpd:durant}, where separate channel filters (one for each
communication link) are deployed to consistently fuse local estimates.
Fig.~\ref{g1} compares the proposed CF/DPF framework with channel filter framework and the centralized Architecture.
In the channel filter framework (Fig.~\ref{g1}(c)), the number of channel filters implemented at each node
equals the number of connections it has with its neighbouring nodes and, therefore, varies from one node to another.
These filters are in addition to the localized filters run at nodes.
In the CF/DPF (Fig.~\ref{g1}(a)) each node only implements one additional filter irrespective of the~neighbouring connections.
Finally, the tree-connect network shown in~Fig.~\ref{g1}(c) can not be extended to any arbitrary network, for example the one  shown in Fig.~\ref{g1}(a). The CF/DPF is applicable to any network~configuration.
\begin{figure}
\centerline{
\includegraphics[scale=0.4]{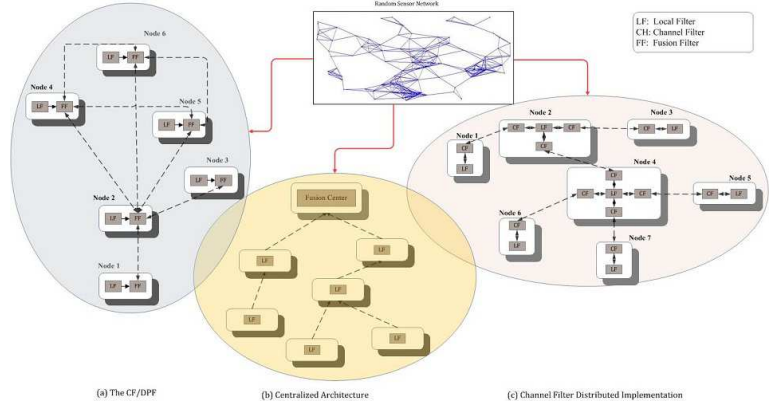}}
\caption{(a) The proposed CF/DPF implementation where sensor nodes connect through their fusion filters (one fusion filter per node). (b) Centralized implementation where all nodes communicate their local estimates to the fusion center. (c) Distributed implementation using channel filters where a separate filter is required for each communication link. In terms of the number of extra filters, the CF/DPF falls between the Centralized and channel~filters.}
\label{g1}
\end{figure}

\noindent
\textbf{2. Modified Fusion filters:}
In the CF/DPF, the fusion filters can run at a rate different form that of the local filters. We further investigate this multi-rate nature of the proposed framework,
recognize three different scenarios, and describe how the CF/DPF handles each of them.
For the worse-case scenario with the fusion filters lagging the local filters exponentially,
we derive a modified-fusion filter algorithm that limits the lag to an affordable delay.

\noindent
\textbf{3. Computing Posterior Cram\'er-Rao Lower Bound:}
In order to evaluate the performance of the proposed \textit{distributed}, \textit{non-linear} framework, we derive the posterior Cram\'er-Rao lower bound (PCRLB), (also referred in literature as the
Bayesian CRLB) for the distributed architecture.
The current PCRLB approaches~\cite{Tichavsky:1998, Hernandez:2002, Zuo:2011} assume a centralized
architecture or a hierarchical architecture~\cite{Tharmarasa:2009}.
 The exact expression for computing the
PCRLB for the distributed architecture
is not yet available and only~an approximate
expression~\cite{Tharmarasa:2011} has recently been derived.
The paper derives the exact expression for computing the
PCRLB for the distributed architecture.
Following Tichavsky \textit{et al.}~\cite{Tichavsky:1998},
we provide a Riccati-type recursion that sequentially determines the
exact FIM from localized FIMs of the distributed estimator.

The rest of paper is organized as follows. Section~\ref{sec:background} introduces notation and reviews the centralized particle filter as well as the average consensus approaches. The proposed CF/DPF algorithm and the fusion filter are described in Section~\ref{sec:gcf/dpf}. The modified fusion filter is presented in
Section~\ref{sec:mff}. Section~\ref{sec:pcrlb} derives an expression for computing the PCRLB for a distributed architecture. Section~\ref{sec:simu} illustrates the effectiveness of the proposed framework in tracking applications
through Monte Carlo simulations. Finally, in Section~\ref{sec:conclusion}, we conclude the paper.
%

\section{BACKGROUND} \label{sec:background}
We consider a sensor network comprising of $N$ nodes observing a set
of $n_x$ state variables $\x = [ X_1, X_2, \ldots , X_{n_x}]^T$. For
($1 \leq l \leq N$), node $l$ makes a measurement $\z^{(l)}(k)$ at
discrete time instants $k$, ($1 \leq k$).  The global observation
vector is given by $\z = [\z^{(1)T}, \ldots, \z^{(N)T}]^T$, where $T$
denotes transposition. The overall state-space representation of the
dynamical system is given by
\begin{eqnarray} \label{sec.3:eq.1.1}
\mbox{State Model:}~~~~~~~~~~~~~~~~~~~~
\x(k) &=& \f(\x(k-1)) + \bm{\xi}(k)\\
\mbox{Observation Model: }
\underbrace{\left [
\begin{array}{c}
\z^{(1)}(k) \\
\vdots \\
\z^{(N)}(k)
\end{array} \right]}_{\z(k)} &=&
\underbrace{\left [
\begin{array}{c}
\bm{g}^{(1)}(\x(k)) \\
\vdots \\
\bm{g}^{(N)}(\x(k))
\end{array} \right ]}_{\g(\x(k))} +
\underbrace{\left [
\begin{array}{c}
\bm{\zeta}^{(1)}(k) \\
\vdots \\
\bm{\zeta}^{(N)}(k)
\end{array} \right],}_{\bm{\zeta}(k)}~~~~~\label{sec.3:eq.1.2}
\end{eqnarray}
where $\bm{\xi}(\cdot)$ and $\bm{\zeta}(\cdot)$ are, respectively, the
global uncertainties in the process and observation models.  Unlike
the Kalman filter, the state and observation functions $\f(\cdot)$ and
$\g(\cdot)$ can possibly be nonlinear, and vectors $\bm{\xi}(\cdot)$
and $\bm{\zeta}(\cdot)$ are not necessarily restricted to white
Gaussian noise.

The optimal Bayesian filtering recursion for iteration $k$ is given by
\begin{eqnarray}\label{ar:b11}
P(\x(k)|\z(1\!:\!k\!-\!1)) &=&  \int P(\x(k\!-\!1)|\z(1\!:\!k\!-\!1))f(\x(k)|\x(k\!-\!1))d\x(k\!-\!1)\\
\mbox{and} ~ P(\x(k)|\z(1\!:\!k)) &=& \frac{P(\z(k)|\x(k))P(\x(k)|\z(1\!:\!k\!-\!1))}{P(\z(k)|\z(1\!:\!k\!-\!1))}.\label{ar:b12}
\end{eqnarray}
The particle filter is based on the principle of sequential importance
sampling~\cite{Tutorial:PF, Bearing:gordon}, a suboptimal technique for
implementing recursive Bayesian estimation (Eqs.~(\ref{ar:b11})
and~(\ref{ar:b12})) through Monte Carlo simulations.
The basic idea behind the particle filters is that the posterior distribution
$P(\x(0\!:\!k)|\z(1\!\!:\!\!k))$ is represented by a collection of weighted random particles
$\{\bm{\cx}_{i}(k)\}_{i=1}^{N_s}$ derived from a proposal distribution
$q(\x(0\!\!:\!\!k)|\z(1\!\!:\!\!k))$ with normalized weights $W_{i}(k)
= \frac{P(\cx_i(k)|\z(1:k))}{q(\cx_i(0:k)|\z(1:k))}$, for $(1\leq
i\leq N_s)$, associated with the vector particles. The particle filter implements
the filtering recursions approximately by propagating the weighted
particles, ($1 \leq i \leq N_s$), using the following recursions
at iteration $k$.
\begin{eqnarray} \label{proposal-1}
\mbox{Time Update:} \quad\quad  \cx_i(k) &\!\!\sim\!\!& q\Big(\bm{\cx}_i(k)|\cx_i(0\!:\!k\!-\!1), \z(1\!:\!k)\Big) \\
\mbox{Observation Update:} \quad\quad W_i(k) &\!\!\propto\!\!& W_i(k-1)\frac{P\Big(\z(k)|\bm{\cx}_i(k)\Big) P\Big(\bm{\cx}_i(k)|\cx_i(k\!-\!1)\Big)}{q\Big(\bm{\cx}_i(k)|\cx_i(0\!:\!k\!-\!1), \z(1\!:\!k)\Big)}, \label{proposal-1-2}
\end{eqnarray}
where $\propto$ stands for the proportional sign and
the proposal distribution satisfies the following~factorization
\begin{eqnarray} \label{sec.3:eq.6}
q\Big(\x(0\!:\!k)|\z(1\!:\!k)\Big) =  q\Big(\x(0\!:\!k\!-\!1)|\z(1\!:\!k\!-\!1)\Big)
q\Big(\x(k)|\x(0\!:\!k\!-\!1), \z(1\!:\!k)\Big).
\end{eqnarray}
The accuracy of this importance sampling approximation depends on how
close the proposal distribution is to the true posterior distribution.
The optimal choice~\cite{Unscented:pk} for the proposal distribution that minimizes the variance of importance weights is the filtering density conditioned upon $\x(0:k-1)$ and $\z(k)$, i.e.,
\begin{eqnarray} \label{sec.3:eq.8}
q\Big(\x(k)|\x(0\!:\!k\!-\!1), \z(1\!:\!k)\Big) = P\Big(\x(k)|\x(0\!:\!k\!-\!1), \z(k)\Big).
 \end{eqnarray}
Because of the difficulty in sampling Eq.~(\ref{sec.3:eq.8}), a common
choice~\cite{Unscented:pk} for the proposal distribution is the
transition density, $P(\x(k)|\x(k-1))$, referred to as the sampling
importance resampling (SIR) filter, where the weights are pointwise
evaluation of the likelihood function at the particle values,~i.e.,
$W_i(k) \propto W_i(k\!-\!1)P(\z(k)|\cx_i(k))$.

\vspace{0.05in}
\noindent
\textit{B. Average Consensus Algorithms} \label{sec:aca} 

The average consensus algorithm~\cite{Olfati:4, Dimakis:2010} considered in the manuscript is represented by
\begin{eqnarray} \label{sec.4:eq.5}
X^{(l)}_c(t + 1) = U_{ll}(t) X^{(l)}_c(t) + \sum_{j \in \aleph^{(l)}} U_{lj}(t) X^{(j)}_c(t),
\end{eqnarray}
where $X^{(l)}_c$(t) is the consensus state variable(s) at node $l$, for ($1 \leq l \leq N$),  $t$ is the consensus time
index that is different from the filtering time index $k$, and $\aleph(l)$ represents the set of neighbouring nodes for
node~$l$. The convergence properties of the average consensus algorithms are reviewed in~\cite{Olshevsky:2009}.  Please refer to~\cite{Olfati:4, Dimakis:2010} for further details on consensus algorithms.
\section{The CF/DPF Implementation} \label{sec:gcf/dpf}
The  CF/DPF implementation runs two localized particle filters at each sensor
 node as shown in Fig.~\ref{g1}.
The first filter, referred to as the local filter, comes from
the distributed implementation of the particle filter described in
Section~\ref{sec:cpf} and is based only on the local observations
$\z^{(l)}(1\!:\!k)$.
The CF/DPF introduces a second particle filter at each node, referred
to as the fusion filter, which estimates the global
posterior distribution $P( \x(0\!:\!k)|\z(1\!:\!k))$ from local distributions $P(\x(k)|\z^{(l)}(1\!:\!k))$ and $P(\x(k)|\z^{(l)}(1\!:\!k\!-\!1))$.
\subsection{Distributed Configuration and Local Filters} \label{sec:cpf}
Our distributed implementation is based on the following model
\begin{eqnarray}
\x(k) &=& \f(\x(k-1)) + \bm{\xi}(k)\\
\z^{(l)}(k) &=&
\bm{g}^{(l)}(\x(k)) + \bm{\zeta}^{(l)}(k),
\end{eqnarray}
for sensor nodes $(1 \leq l \leq N)$. In other words, the entire state
vector $\x(k)$ is estimated by running one localized particle filter at each node.
These filters, referred to as the local filters, come from the
distributed implementation of the particle filter and are based only on local
observations $\z^{(l)}(1\!:\!k)$.  In addition to updating the
particles and their associated weights, the local filter at node $l$
provides estimates of the local prediction distribution $P(
\x(k)|\z^{(l)}(1\!:\!k\!-\!1))$ from the particles as explained below.
\vspace*{0.1in}
\\\textbf{Computation and Sampling of the Prediction Distribution:}
From the Chapman-Kolmogorov equation (Eq.~(\ref{ar:b11})), a sample based approximation of the prediction density $P(\x(k)|\z^{(l)}(1\!:\!k\!-\!1))$ is expressed as
\begin{eqnarray} \label{chap5:sec3:part1:eq13}
P\left(\x(k)|\z^{(l)}(1\!:\!k\!\!-\!\!1)\right) = \sum_{i=1}^{N_s} W^{(l)}_i(k\!\!-\!\!1)P\left(\x(k)|\cx^{(l)}_i(k\!\!-\!\!1)\right),
\end{eqnarray}
which is a continuous mixture.  To generate random particles
 from such a mixture density, a new sample $\cx^{(l)}_i(k|k\!-\!1)$ is
 generated from its corresponding mixture
 $P(\x(k)|\cx^{(l)}_i(k\!-\!1))$ in Eq.~(\ref{chap5:sec3:part1:eq13}).
 Its weight $W^{(l)}_i(k\!-\!1)$ is the same as the corresponding
 weight for $\cx^{(l)}_i(k\!-\!1)$. The prediction  density is given by
\begin{equation}
P\left(\x(k)|\z^{(l)}(1\!:\!k\!-\!1)\right) = \sum_{i=1}^{N_s} W^{(l)}_i(k\!-\!1)
\delta\!\left( \x(k)\!-\!\cx^{(l)}_i(k|k\!-\!1)\right).\nonumber
\end{equation}
Once the random samples are generated, the minimum mean square error estimates (MMSE) of the
parameters can be computed.


\vspace{.05in}
\noindent
\textit{B. Fusion Filter} \label{sec:ff}

The CF/DPF introduces a second particle filter at each node, referred
to as the fusion filter, which computes an estimate of the global
posterior distribution $P( \x(0\!:\!k)|\z(1\!:\!k))$.
Being a particle filter itself, implementation of the fusion filter requires the proposal distribution and the weight update equation.
Theorem~\ref{fprotocol}~\cite{Chong:1990} expresses the global posterior distribution in terms of the local filtering densities, which is used for updating the weights of the fusion filter. The selection of the proposal distribution is explained later in Section~III-E.
In the following discussion, the fusion filter's particles and their associated
weights at node $l$ are denoted by $ \{ \cx^{(l, \GF)}_i(k), W^{(l, \GF)}_i(k)
\}_{i=1}^{N_{\GF}}$.
\begin{thm}\label{fprotocol}
Assuming that the observations conditioned on the state variables and
made at node $l$ are independent of those made at node $j$,
($j \neq l$), the global posterior distribution for a $N$--sensor
network~is 
\begin{eqnarray} \label{chap3:sec2:eq8}
P\Big(\x(0\!:\!k)|\z(1\!:\!k)\Big) \propto \frac{\prod_{l=1}^N P
\Big( \x(k)|\z^{(l)}(1\!:\!k)\Big)}{\prod_{l=1}^N P
\Big( \x(k)|\z^{(l)}(1\!:\!k\!-\!1)\Big)}
\!\times\! P\Big(\x(0\!:\!k)|\z(1\!:\!k\!-\!1)\Big),
\end{eqnarray}
where the last term can be factorized as follows
\begin{eqnarray} \label{chap5:sec3:part1:eq2}
 P\Big(\x(0\!:\!k)|\z(1\!:\!k\!-\!1)\Big) =  P\Big(\x(k)|\x(k\!-\!1)\Big)
P\Big(\x(0\!:\!k\!-\!1)|\z(1\!:\!k\!-\!1)\Big).
\end{eqnarray}
\end{thm}
The proof of Theorem~\ref{fprotocol} is included in Appendix~\ref{app:A}.
Note that the optimal distributed protocol defined in
Eq.~(\ref{chap3:sec2:eq8}) consists of three terms: (i) Product of the
local filtering distribution $\prod_{i=1}^N P(
\x(k)|\z^{(l)}(1\!:\!k))$ which depends on local
observations; (ii) Product of local prediction densities
$\prod_{i=1}^N P( \x(k)|\z^{(l)}(1\!:\!k\!-\!1))$, which is
again only based on the local observations and represent the common
information between neighboring nodes, and; (iii) Global prediction
density $P\left( \x(0\!:\!k)|\z(1\!:\!k\!-\!1)\right)$ based on
Eq.~(\ref{chap5:sec3:part1:eq2}). The fusion rule, therefore, requires
consensus algorithms to be run for terms (i) and (ii).
The proposed CF/DPF computes the two terms separately (as described later) by running two consensus algorithms at each iteration of the fusion filter. An alternative is to compute the ratio (i.e., proportional to the local likelihood) at each node and run one consensus algorithm for computing the ratio term.
In the CF/DPF, we propose to estimate the numerator and denominator of Eq.~\eqref{chap3:sec2:eq8} separately  because
 maintaining the local filtering and prediction distributions is advantageous in networks with intermittent connectivity as it allows the CF/DPF to recover from loss of information due to delays in convergence.
Maintaining the likelihood prevents the recovery of the CF/DPF in such cases.

\vspace{.05in}
\noindent
\textit{C. Weight Update Equation}

Assume that the
local filters have reached steady state at iteration $k$, i.e., the
local filter's computation is completed up to and including time
iteration $k$ where a particle filter based estimate of the local
filtering distribution is available. The weight update equation for
the fusion filter is given by
\begin{eqnarray} \label{chap5:sec3:part1:eq4}
W^{(l,\GF)}_i(k) = \frac{ P\left(\cx^{(l,\GF)}_i(k)|\z(1\!:\!k)\right)}
{q\left(\cx^{(l,\GF)}_i(k)|\z(1\!:\!k)\right)}.
\end{eqnarray}
The CF/DPF is derived based on the global posterior $P(\x(0\!:\!k)|\z(1\!:\!k))$ which is the standard approach in the particle filter literature~\cite{Tutorial:PF}. Further, we are only interested in a filtered estimate of the state variables
$P(\x(k)|\z(1\!:\!k))$ at each iteration. Following~\cite{Tutorial:PF} we, therefore,  approximate $q(\x(k)|\x(1\!:\!k\!-\!1), \z(1\!:\!k)) =
q(\x(k)|\x(k\!-\!1), \z(k))$.
 The proposal density is then dependent
only on $\x(k)$ and $\z(k)$. In such a scenario, one can discard the
history of the
particles $\cx^{(l,\GF)}_i(0\!:\!k\!-\!2)$ at previous
iterations~\cite{Tutorial:PF}.  Substituting
Eq.~(\ref{chap5:sec3:part1:eq2}) in Eq.~(\ref{chap3:sec2:eq8}) and
using the result together with Eq.~(\ref{sec.3:eq.6}) in
Eq.~(\ref{chap5:sec3:part1:eq4}), the weight update equation is given by
\begin{eqnarray} \label{chap5:sec3:part1:eq5}
\lefteqn{W^{(l,\GF)}_i(k)\propto W^{(l,\GF)}_i(k\!-\!1)
\frac{\prod_{l=1}^N P\left( \cx^{(l,\GF)}_i(k)|\z^{(l)}(1\!:\!k)\right)}
{\prod_{l=1}^N P\left( \cx^{(l,\GF)}_i(k)|\z^{(l)}(1\!:\!k\!-\!1)\right)}\;
\frac{P\left(\cx^{(l,\GF)}_i(k)|\cx^{(l,\GF)}_i(k\!-\!1)\right)}
{q\left(\cx^{(l,\GF)}_i(k)|\cx^{(l,\GF)}_i(k\!-\!1), \z(k)\right)}, }\\
&&\!\!\!\!\!\!\!\!\!\!\mbox{where} \quad\quad\quad\quad\quad\quad\quad\quad W^{(l,\GF)}_i(k\!-\!1) = \frac{P\left(\cx^{(l,\GF)}_i(k\!-\!1)|\z(1\!:\!k\!-\!1)\right)}{q\left(\cx^{(l,\GF)}_i(k-1)|\z(1\!:\!k\!-\!1)\right)}.~~~~~~~~~~~~~~~~~~~~~~~~~~~~~~\label{chap5:sec3:part1:eq6}
\end{eqnarray}
Observe that only the first fraction in
Eq.~(\ref{chap5:sec3:part1:eq5}) requires all nodes to participate.
Given the weights $W^{(l,\GF)}_i(k\!-\!1)$ from the previous iteration,
Eq.~(\ref{chap5:sec3:part1:eq5}) requires the following distributions
\begin{eqnarray}\label{term1}
\prod_{l=1}^N P\left(
\cx^{(l,\GF)}_i(k)|\z^{(l)}(1\!:\!k)\right) \quad\quad
\mbox{and} \quad\quad \prod_{l=1}^N P\left(
\cx^{(l,\GF)}_i(k)|\z^{(l)}(1\!:\!k\!-\!1)\right).
\label{term2}
\end{eqnarray}
The numerator of the second fraction in
Eq.~(\ref{chap5:sec3:part1:eq5}) requires the transitional
distribution $P(\x(k)|\x(k-1))$, which is known from the state
model. Its denominator requires the proposal distribution
$q(\x(k)|\x(k\!-\!1),\z(k))$.  Below, we show how the three terms (Eq.~(\ref{term1}) and the proposal distribution) are
determined.
%

\vspace{.05in}
\noindent
\textit{D. Distributed Computation of Product Densities}

Distributions in Eq.~(\ref{term1}) are not
determined by transferring the whole particle vectors and their
associated weights between the neighboring nodes due to an
impractically large number of information transfers.  Further, the
localized posteriors are represented as a Dirac mixture
 in the particle filter.  Two separate Dirac
mixtures may not have the same support and their multiplication could
possibly be zero. If not, the product may not represent the true
product density accurately. In order to tackle these problems, a
transformation is required on the Dirac function particle
representations by converting them to continuous distributions prior
to communication and fusion.
 Gaussian distributions~\cite{Rigatos:2010,Simonetto:2010,Huang:2008,consdpf:1,Coates:2010,Arash:ssp2011}, grid-based techniques~\cite{Rosencrantz:2003}, Gaussian Mixture Model (GMM)~\cite{sheng:2005} and Parzen representations~\cite{dpd:durant} are different parametric continuous distributions used in the context of the distributed particle filter implementations.
The channel filter framework~\cite{dpd:durant} fuses only two local distributions, therefore, the local pdfs can be modeled~\cite{dpd:durant} with such complex distributions.
Incorporating these distributions in the CF/DPF framework is, however, not a trivial task because the CF/DPF computes the product of $N$ local distributions. The use of a complex distribution like GMM is, therefore,  computationally prohibitive.

In order to tackle this problem, we approximate the product terms in
Eq.~(\ref{chap5:sec3:part1:eq5}) with Gaussian distribution which
results in local filtering and prediction densities to be normally
distributed as follows
\begin{eqnarray} \label{chap5:sec3:part1:eq7}
P\left(\x(k)|\z^{(l)}(1\!:\!k)\right) \propto
\mathcal{N}\left(\mu^{(l)}(k),\bm{P}^{(l)}(k)\right) \mbox{ and }
P\left(\x(k)|\z^{(l)}(1\!:\!k-1)\right) \propto
\mathcal{N}\left(\bm{\nu}^{(l)}(k),\bm{R}^{(l)}(k)\right),
\end{eqnarray}
where $\bm{\mu}^{(l)}(k)$ and $\bm{P}^{(l)}(k)$ are, respectively, the
mean and covariance of local particles at node $l$ during the
filtering step of iteration $k$. Similarly, $\bm{\nu}^{(l)}(k)$ and
$\bm{R}^{(l)}(k)$ are, respectively, the mean and covariance of local
particles at node $l$ during the prediction step. It should be noted
that we only approximate the product density for updating the weights
with a Gaussian distribution and the global posterior distribution is
not restricted to be Gaussian.
The local statistics at node $l$ are computed as
\begin{eqnarray} \label{chap5:sec3:part1:eq14-1}
\bm{\mu}^{(l)}(k) \!\!=\!\! \sum_{i=1}^{N_s} W^{(l)}_i(k)\cx^{(l)}_i(k)
~~\mbox{and }~~ \bm{P}^{(l)}(k)\!\!=\!\!
\sum_{i=1}^{N_s} W^{(l)}_i(k)\left( \cx^{(l)}_i(k) \!-\! \bm{\mu}^{(l)}(k) \right) \left( \cx^{(l)}_i(k) \!-\! \bm{\mu}^{(l)}(k) \right)^T\!\!\!\!\!.
\end{eqnarray}
Reference~\cite{Gales:2006} shows that
the product of $N$ multivariate normal distributions is also normal, i.e.,
\begin{eqnarray} \label{chap5:sec3:part1:eq8}
\prod_{l=1}^N P\left( \x(k)|\z^{(l)}(1\!:\!k)\right) \triangleq
\prod_{l=1}^N \mathcal{N}\left(\mu^{(l)}(k),\bm{P}^{(l)}(k)\right)
=  \frac{1}{C} \times \mathcal{N}(\bm{\mu}(k),\bm{P}(k)),
\end{eqnarray}
where $C$ is a normalization term (Reference~\cite{Gales:2006} includes the proof).
Parameters $\bm{\mu}(k)$ and $\bm{P}(k)$ are 
\begin{eqnarray} \label{chap5:sec3:part1:eq9-1}
\bm{P}(k) = \big( \sum_{l=1}^N
\underbrace{\left(\bm{P}^{(l)}(k)\right)^{-1}}_{\X^{(l)}_{c1}(0)} \big)^{-1}
\quad \mbox{and} \quad
\bm{\mu}(k) = \bm{P}(k)\times \sum_{l=1}^N
\underbrace{\left(\bm{P}^{(l)}(k)\right)^{-1}\bm{\mu}^{(l)}(k).}_{\x^{(l)}_{c2}(0)}
\end{eqnarray}
Similarly, the product of local prediction densities (Term~(\ref{term2}))
is modeled with a Gaussian density
$\mathcal{N}(\x(k);\bm{\upsilon}(k),\bm{R}(k))$, where
the parameters $\bm{\upsilon}(k)$ and
$\bm{R}(k)$ are computed as follows
\begin{eqnarray} \label{chap5:sec3:part1:eq9-12}
\bm{R}(k) = \big( \sum_{l=1}^N
\underbrace{\left(\bm{R}^{(l)}(k)\right)^{-1}}_{\X^{(l)}_{c3}(0)} \big)^{-1}
\quad \mbox{and} \quad
\bm{\upsilon}(k) = \bm{R}(k)\times \sum_{l=1}^N
\underbrace{\left(\bm{R}^{(l)}(k)\right)^{-1}\bm{\upsilon}^{(l)}(k).}_{\x^{(l)}_{c4}(0)}
\end{eqnarray}
 The parameters of the product distributions only involves average quantities and can be
provided using average consensus algorithms as follows:
\begin{itemize}
\item[(i)] For ($1\!\leq\! l\!\leq\! N$), node $l$ initializes its consensus
  states to
$\X^{(l)}_{c1}(0) = (\bm{P}^{(l)}(k))^{-1}$,
$\x^{(l)}_{c2}(0) = (\bm{P}^{(l)}(k))^{-1}\bm{\mu}^{(l)}(k)$,
$\X^{(l)}_{c3}(0) = (\bm{R}^{(l)}(k))^{-1}$,
and $\x^{(l)}_{c4}(0) = (\bm{R}^{(l)}(k))^{-1}\bm{\upsilon}^{(l)}(k)$,
and then Eq.~(\ref{sec.4:eq.5}) is used to reach consensus with $\X^{(l)}_{c1}(t)$ used instead of $X_c^{(l)}(t)$ in Eq.~(\ref{sec.4:eq.5}) for the first consensus run. Similarly, $\x^{(l)}_{c2}(t)$ is used instead of $X_c^{(l)}(t)$ for the second run and so on.
\item[(ii)] Once consensus is reached, parameters $\bm{\mu}^{(l)}(k)$
  and $\bm{P}^{(l)}(k)$ are computed as follows
\begin{eqnarray}
\label{chap5:sec3:part1:eq16}
\bm{P}(k) &\!=\!& 1/N \times \lim_{t\rightarrow\infty}
\left\{\left(\X^{(l)}_{c1}(t)\right)^{-1} \right\} \quad \mbox{and} \quad
\bm{\mu}(k) \!=\!
\lim_{t\rightarrow\infty} \left\{ \left(\X^{(l)}_{c1}(t)\right)^{-1} \times
\x^{(l)}_{c2}(t) \right\} \\
\bm{R}(k) &\!=\!& 1/N \times \lim_{t\rightarrow\infty}
\left\{\left(\X^{(l)}_{c3}(t)\right)^{-1} \right\} \quad \mbox{and} \quad
\bm{\upsilon}(k) \!=\!
\lim_{t\rightarrow\infty} \left\{ \left(\X^{(l)}_{c3}(t)\right)^{-1} \times
\x^{(l)}_{c4}(t) \right\}.
\end{eqnarray}
\end{itemize}
Based on aforementioned approximation, the weight update equation of
the fusion filter (Eq.~(\ref{chap5:sec3:part1:eq5})) is 
\begin{eqnarray} \label{chap5:sec3:part1:eq15}
W^{(l,\GF)}_i(k) \!\propto\! W^{(l,\GF)}_i(k\!-\!1)
\frac{\mathcal{N}\big(\cx^{(l,\GF)}_i (k); \bm{\mu}(k),\bm{P}(k)\big)P\big(\cx^{(l,\GF)}_i(k)|\cx^{(l,\GF)}_i(k\!-\!1)\big)}
{\mathcal{N}\big(\cx^{(l,\GF)}_i  (k); \bm{\upsilon}(k),\bm{R}(k)\big)q\big(\cx^{(l,\GF)}_i(k)|\cx^{(l,\GF)}_i(k\!-\!1), \z(k)\big)}.
\end{eqnarray}
Eq.~(\ref{chap5:sec3:part1:eq15}) requires the proposal
distribution $q(\x(k)|\x(k\!-\!1), \z(k))$, which is discussed next.

\vspace{.05in}
\noindent
\textit{E. Proposal Distribution} \label{sec:prop}

In this section, we describe three different proposal distributions in CF/DPF.

\noindent
\textbf{$1.$ SIR Fusion Filter:}
The most common strategy is to sample from the probabilistic model of the state evolution, i.e., to use transitional density $P(\x(k)|\x(k\!+\!1))$ as proposal distribution. The simplified weight update equation for the SIR fusion filter is obtained from Eq.~(\ref{chap5:sec3:part1:eq15}) as follows
\begin{eqnarray} \label{chap5:sec3:part1:eq155}
W^{(l,\GF)}_i(k) \!\propto\! W^{(l,\GF)}_i(k\!-\!1)
\frac{\mathcal{N}\big(\cx^{(l,\GF)}_i (k); \bm{\mu}(k),\bm{P}(k)\big)}
{\mathcal{N}\big(\cx^{(l,\GF)}_i  (k); \bm{\upsilon}(k),\bm{R}(k)\big)}.
\end{eqnarray}
This SIR fusion filter fails if a new measurement appears in the tail of the transitional distribution or when the likelihood is too peaked in comparison with the transitional density.

\noindent
\textbf{$2.$ Product Density as Proposal Distribution:}
We are free to choose any proposal distribution that
appropriately considers the effect of new observations and is close to
the global posterior distribution. The product of local
filtering densities is a reasonable approximation of the
global posterior density as such a good candidate for the proposal
distribution, i.e.,
\begin{eqnarray} \label{chap5:sec3:part1:eq10}
q(\x(k)|\x(k\!-\!1), \z(1\!:\!k)) \triangleq
\prod_{l=1}^N P\left( \x(k)|\z^{(l)}(1\!:\!k) \right),
\end{eqnarray}
which means that we generate particles
$\{\cx^{(l,\GF)}_i(k)\}_{i=1}^{N_s}$ are generated from
$\mathcal{N}(\bm{\mu}(k),\bm{P}(k))$. In such a scenario, the weight
update equation (Eq.~(\ref{chap5:sec3:part1:eq15})) simplifies to
\begin{eqnarray} \label{chap5:sec3:part1:eq11}
W^{(l,\GF)}_i(k)\!\propto\! W^{(l,\GF)}_i(k\!-\!1)
\frac{P(\cx^{(l,\GF)}_i(k)|\cx^{(l,\GF)}_i(k\!-\!1))}
{\mathcal{N}(\cx^{(l,\GF)}_i(k) ; \bm{\upsilon}(k),\bm{R}(k))}.
\end{eqnarray}
Next we justify that the product term is a good choice and a near-optimal approximation of the optimal proposal distribution (Eq.~(\ref{sec.3:eq.8})).
 Assume at iteration $k$, node $l$, for ($1 \leq l \leq N$) computes an unbiased local estimate $\hat{\x}^{(l)}(k)$ of the state variables $\x(k)$ from its particle-based representation of the filtering distribution
with the corresponding error and error covariance denoted by $\Delta_x^{(l)}(k) = \x(k)-\hat{\x}^{(l)}(k)$ and $\hat{\bm{P}}^{(l)}(k)$.
When the estimation error $\Delta_x^{(i)}(k)$ and $\Delta_x^{(j)}(k)$, for ($1 \leq i,j \leq N$) and $i \neq j$ are uncorrelated, the optimal fusion of $N$ unbiased local estimates $\hat{\x}^{(l)}(k)$ in linear minimum variance scene is shown~\cite{sun:2004} to be given by
\begin{eqnarray} \label{ext:CF/DPF:eq4}
\hat{\bm{P}}(k) = \big( \sum_{l=1}^N \left(\hat{\bm{P}}^{(l)}(k)\right)^{-1} \big)^{-1}
\quad \mbox{and} \quad
\hat{\x}(k) = \big( \sum_{l=1}^N
\left(\hat{\bm{P}}^{(l)}(k)\right)^{-1} \big)^{-1} \times \sum_{l=1}^N
\left(\hat{\bm{P}}^{(l)}(k)\right)^{-1}\hat{\x}^{(l)}(k).
\end{eqnarray}
where $\hat{\x}(k)$ is the overall estimate obtained from $P( \x(k)|\z(1\!:\!k))$
with error covariance $\hat{\bm{P}}(k)$.
Eq.~(\ref{ext:CF/DPF:eq4})
 is the same as Eq.~(\ref{chap5:sec3:part1:eq9-1}), which describes the statistics of the product of $N$ normally distributed densities.
The optimal proposal distribution is also a filtering density~\cite{Tutorial:PF},
therefore, the proposal distribution defined in Eq.~(\ref{chap5:sec3:part1:eq10}) is a good choice that simplifies the update equation of the fusion filter. Further, Eq.~(\ref{chap5:sec3:part1:eq10}) is a reasonable approximation of the optimal proposal distribution.
From the framework of unscented Kalman filter and unscented particle filter, it is well known~\cite{Unscented:pk} that approximating distributions will be advantageous over approximating non-linear functions.
The drawback with this proposal density is the impractical assumption that the local estimates are uncorrelated. We improve the performance of the fusion filter using a better approximation of the optimal proposal distribution, which is described next.

\noindent
\textbf{$3.$ Gaussian Approximation of The Optimal Proposal Distribution:}
We consider the optimal solution to the fusion protocol (Eq.~(\ref{chap3:sec2:eq8})) when local filtering densities are normally distributed. In such a case,  $P(\x(0\!:\!k)|\z(1\!:\!k-1))$ is also normally distributed~\cite{Chong:1990}
with mean  $\x^{(l, \ggg)}(k)$ and
covariance $\bm{P}^{(l, \ggg)}(k)$
\begin{eqnarray} \label{eq1-1:sec4}
\!\!\!\!\!\!\!\!\!\!\!\!\!\!\!\!\!\!\!&& \hat{\bm{P}}^{(l, \ggg)^{-1}}(k) = \left(\bm{R}^{(l)}(k)\right)^{-1} +
\underbrace{\sum_{j=1}^{N} \bm{P}^{(j)^{-1}}(k)}_{\x^{(l)}_{c1}(\infty)} -\underbrace{\sum_{j=1}^{N}
\bm{R}^{(j)^{-1}}(k)}_{\x^{(l)}_{c3}(\infty)} \\
\!\!\!\!\!\!\!\!\!\!\!\!\!\!\!\!\!\!\!&& \hat{\x}^{(l, \ggg)}(k) = \bm{P}^{(l, \ggg)^{-1}}(k)
\big[ \left(\bm{R}^{(l)}(k)\right)^{-1} \bm{\upsilon}^{(l)}(k)\!+\! \underbrace{\sum_{j=1}^{N} \bm{P}^{(j)^{-1}}(k)
\bm{\mu}^{(j)}(k)}_{\x^{(l)}_{c2}(\infty)} - \underbrace{\sum_{j=1}^{N}\! \left(\bm{R}^{(j)}(k)\right)^{-1} \bm{\upsilon}^{(j)}(k)}_{\x^{(l)}_{c4}(\infty)}\big].  \label{eq1-2:sec4}
\end{eqnarray}
The four terms $\x_{c1}^{(l)}(\infty)$, $\x_{c2}^{(l)}(\infty)$, $\x_{c3}^{(l)}(\infty)$, and $\x_{c4}^{(l)}(\infty)$ are already computed and available at local nodes as part of computing the product terms.
Fusion rules in Eqs.~(\ref{eq1-1:sec4}) and (\ref{eq1-2:sec4}) are obtained based on the track fusion without feedback~\cite{Chong:1990}.
 In such a scenario, particles $\cx^{(l,\GF)}_i (k)$ are drawn from $\mathcal{N}(\x^{(l, \ggg)}(k),\bm{P}^{(l, \ggg)}(k))$ and the weight
update equation (Eq.~(\ref{chap5:sec3:part1:eq11})) is given by
\begin{eqnarray} \label{chap5:sec3:part1:eq255}
W^{(l,\GF)}_i(k) \!\propto\! W^{(l,\GF)}_i(k\!-\!1)
\frac{\mathcal{N}\big(\cx^{(l,\GF)}_i (k); \bm{\mu}(k),\bm{P}(k)\big)P\big(\cx^{(l,\GF)}_i(k)|\cx^{(l,\GF)}_i(k\!-\!1)\big)}
{\mathcal{N}\big(\cx^{(l,\GF)}_i  (k); \bm{\upsilon}(k),\bm{R}(k)\big)\mathcal{N}\big(\cx^{(l,\GF)}_i  (k); \x^{(l, \ggg)}_{(k)},\bm{P}^{(l, \ggg)}_{(k)}\big)}.
\end{eqnarray}
 The various steps of the fusion filter are outlined in Algorithm~\ref{algo:fusionfilter}.
The filtering step of the CD/DPF is based on running the localized filters at each node followed by the fusion filter,
 which computes the global posterior density by running consensus algorithm across the network. At the completion of the consensus step, all nodes have the same global posterior available.
As a side note to our discussion, we note that the CF/DPF does not incorporate any feedback from the fusion filters to the localized filters to provide sufficient time for the fusion filter to converge.
The main advantage of the feedback is to reduce the error of the local filters which will be considered as~future work.
Finally, a possible future extension of the CF/DPF is to use non-parametric models, e.g., support vector machines (SVM)~\cite{Challa:2002, Liu:2009}, instead,  for approximating the product terms.
An important task in CF/DPF is to assure that the localized and fusion filters do not lose synchronization. This issue is addressed in Section~\ref{sec:mff}.
%
\begin{algorithm*}[t!]
\caption{\textproc{~Fusion Filter}($ \{ \cx^{(l,\GF)}_i(k-1),
  W^{(l,\GF)}_i(k-1) \}_{i=1}^{N_{\GF}}$)}
\label{algo:fusionfilter}

\begin{algorithmic}[1]
\Input $\{\cx^{(l,\GF)}_i(k-1), W^{(l,\GF)}_i(k-1)\}_{i=1}^{N_{l,\GF}}$ -
Fusion filter's particles and associated weights.
\Output $\{\cx^{(l,\GF)}_i(k), W^{(l,\GF)}_i(k)\}_{i=1}^{N_s}$ Fusion
filter's updated particles and associated weights.

\For{$l=1:N,$}

$\left(\bm{\mu}^{(l)}(k),
\bm{P}^{(l)}(k),\bm{\upsilon}^{(l)}(k),\bm{R}^{(l)}(k) \right) = $
LocalFilter$\left(\{ \cx_i^{(l)}(k-1), W_i^{(l)}(k-1)\}_{i=1}^{N_s},
\z^{(l)}(k)\right)$ \EndFor

\State DoFusion$\left(\{ \bm{\mu}^{(l)}(k),
\bm{P}^{(l)}(k)\}_{l=1}^{N}\right)$ computes $\{\bm{\mu}^{(l,\GF)}(k),
\bm{P}^{(l,\GF)}(k)\}$ for numerator of Eq.~(\ref{chap3:sec2:eq8}).

\State DoFusion$\left(\{
\bm{\upsilon}^{(l)}(k),\bm{R}^{(l)}(k)\}_{l=1}^{N}\right)$ computes
$\{\bm{\upsilon}^{(l,\GF)}(k),\bm{R}^{(l,\GF)}(k)\}$ for denominator
of~(\ref{chap3:sec2:eq8}).

\For{$i=1:N,$}

$\bullet$ Generate particles
$\left\{\cx^{(l,\GF)}_i(k)\right\}_{i=1}^{N_{l,\GF}}$ by sampling
proposal distribution defined in Section~III-E. 

$\bullet$ Compute weights $W^{(l,\GF)}(k)$ using
Eq.~(\ref{chap5:sec3:part1:eq11}).  \EndFor

\State If
degeneracy observed
$\left(\{ \cx_i^{(l,\GF)}(k),
W_i^{(l,\GF)}(k)\}_{i=1}^{N_{\GF}}\right)$ = Resample$\left(\{
\cx_i^{(l,\GF)}(k), W_i^{(l,\GF)}(k)\}_{i=1}^{N_{\GF}}\right)$.
\end{algorithmic}
\end{algorithm*}
%

\vspace{.05in}
\noindent
\textit{F. Computational complexity} \label{sec:cc}

In this section, we provide a rough comparison of the computational complexity of the CF/DPF versus that of the centralized implementation.
 Because of the non-linear dynamics of the particle filter, it is somewhat difficult to
derive a generalized expression for its computational complexity.
There are steps that can not be easily evaluated in the complexity computation of the particle filter such as the cost of evaluating a non-linear function (as is the case for the state and observation models)~\cite{Karlsson:2005}.
In order to provide a rough comparison, we consider below a simplified linear state model with Gaussian excitation and  uncorrelated Gaussian observations.
Following the approach proposed in~\cite{Karlsson:2005}, the computational complexity of two implementations of the particle filter is expressed in terms of flops, where a flop is defined as addition, subtraction, multiplication or division of two floating point numbers.
The  computational complexity of the centralized particle filter for $N$--node network with $N_s$ particles is of
 $\mbox{O}\left((n_x^2+N)N_s\right)$.
The CF/DPF runs the  local filter at each observation node which is similar in complexity to the centralized particle filter except that the observation (target's bearing at each node) is a scalar. Setting $N=1$, the computational complexity of the local filter  is of $\mbox{O}\left(n_x^2N_{\text{LF}}\right)$ per node, where $N_{\text{LF}}$ is the number of particles used by the local filter.
There are two additional components in the CF/DPF: (i) The fusion filter which has a complexity of $\mbox{O}(n_x^2N_{\GF})$ per node  where $N_{\GF}$ is the number of particles used by the fusion filter, and;
(ii) The CF/DPF introduces an
additional consensus step which has a computational complexity of $\mbox{O}(n_x^2\Delta_{\mathcal{G}}N_c(\bm{U}))$.
The associated convergence time $N_c(\bm{U}) = 1/\log(1/ r_{\asym}(\bm{U}))$
 provides the asymptotic number of consensus iterations required for the
error to decrease by the factor of $1/e$ and is expressed in terms of the asymptotic convergence rate  $r_{\asym}(\bm{U})$. Based  on~\cite{Olshevsky:2009}, $N_c(\bm{U}) = - 1/ \max_{2 \leq i \leq N} \log(|\lambda_i(\bm{U})|)$, where $\lambda_i(\bm{U})$ is the eigenvalue of the consensus matrix $\bm{U}$.
The overall computational complexity of the CF/DPF is, therefore, given by $\max \big\{\mbox{O}(Nn_x^2(N_{\text{LF}}+N_{\text{\GF}})), \mbox{O}(n_x^2\Delta_{\mathcal{G}}N_c(\bm{U}))\big\}$ compared to the computational complexity $\mbox{O}\left((n_x^2+N)N_s\right)$ of the centralized implementation.
Since the computational complexity of the two
implementations involve different variables, it is difficult to
compare them subjectively.  In our simulations (explained in
Section~\ref{sec:simu}), the value of the variables are: $n_x=4$,
$N=20$, $N_s=10,000$, $N_{\text{LF}} = N_{\text{FF}} =
 500$, and $N_c(\bm{U}) = 8$ for the network in
Fig.~\ref{Target_track}(a) which results in the following rough
computational counts for the two implementations: Centralized
implementation: $3.6\times10^5$, and CF/DPF: $3.4\times10^5$.
 This means that the
two implementations have roughly the same computational complexity
for the BOT simulation of interest to us.  We also note that the
computational burden is distributed evenly across the nodes in the
CF/DPF, while the fusion center performs most of the
computations in the centralized particle filter. This places an
additional power energy constraint on the fusion center causing the
system to fail if the power in the fusion center drains out.

\section{Modified Fusion Filter } \label{sec:mff}
In the CF/DPF, the local filters and the fusion filters can run out of synchronization due to intermittent network connectivity.
The local filters are confined to their sensor node and unaffected by loss of connectivity. The fusion filters, on the other hand,
run consensus algorithms. The convergence of these consensus algorithms is delayed if the communication bandwidth reduced.
In this section we develop ways of dealing with such intermittent connectivity issues.
\begin{figure}
\centerline{
\includegraphics[scale=0.5]{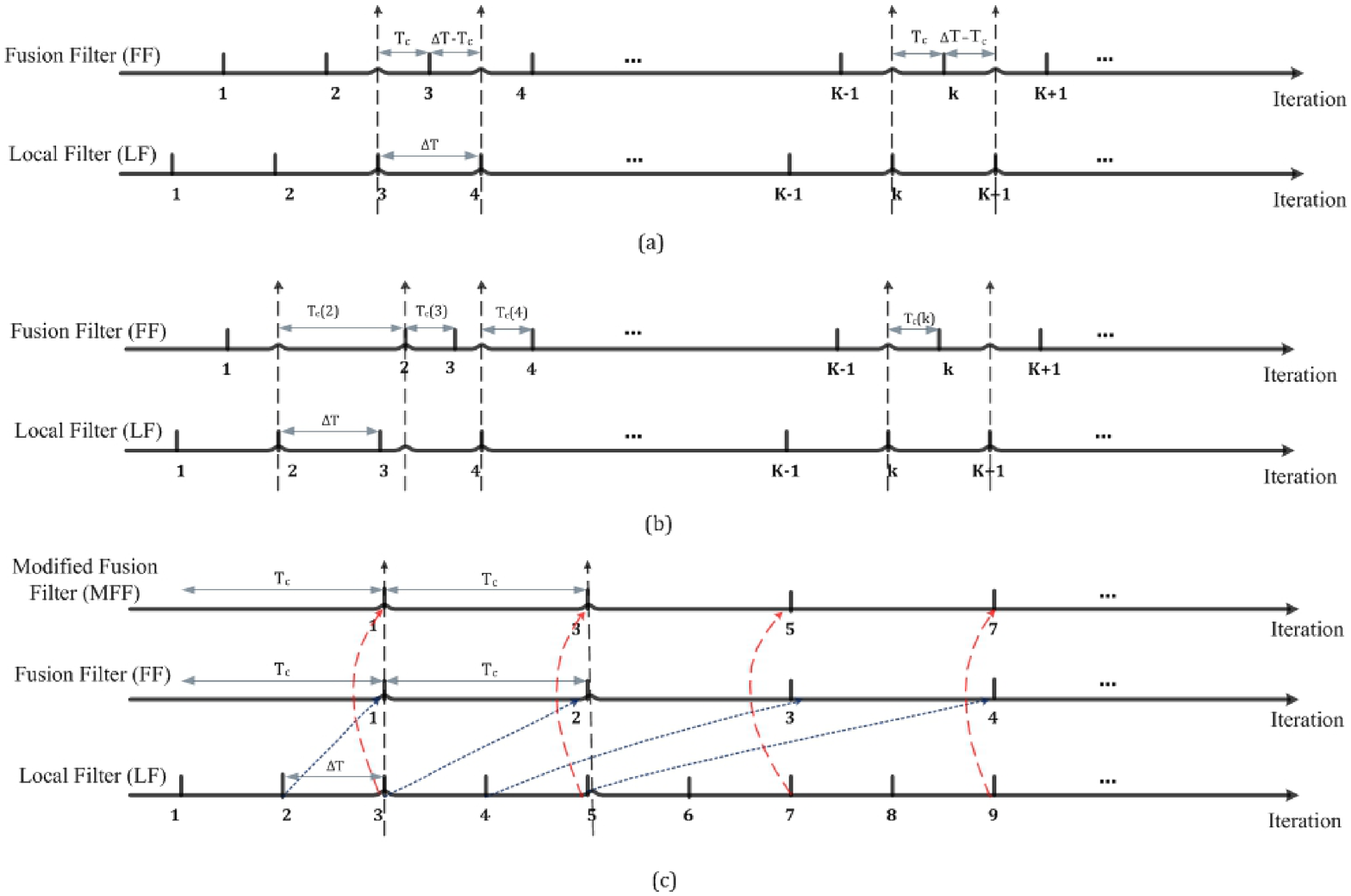}}
\caption{
Multi-rate implementation of the local and fusion filters.
(a) The ideal scenario where the fusion filter's consensus step converges before the new iteration of the local filter. (b)  The convergence rate of the fusion filter varies according to the network connectivity. (c) The lag between the fusion filter and the local filter grows exponentially. }
\label{mff}
\end{figure}
 First, let us introduce the notation. We assume that the observations arrive at constant time intervals of $\Delta T$. Each iteration of the local filters is performed within this interval, which we will refer to as the local filter's estimation interval.
The duration (the fusion filters's estimation interval) of the  update cycle of the fusion filter is denoted by $T_c$.
Fig.~\ref{mff} illustrates three scenarios dealing with different fusion filter's estimation intervals.
Fig.~\ref{mff}(a) is the ideal scenario where $T_c \leq \Delta T$ and the fusion filter's consensus step converges before the new iteration of the local filter.
In such a scenario, the local and fusion filters stay synchronized.
Fig.~\ref{mff}(b) shows the second scenario when the convergence rate of the fusion filter varies according to the network connectivity.
Under regular connectivity $T_c < \Delta T$ and limited connectivity loesses, the fusion filters will manage to catch up with the localized filters in due time.
Fig.~\ref{mff}(c) considers a more problematic scenario when $T_c > \Delta T$. Even with ideal connectivity, the fusion filter will  continue to lag the localized filters with no hope of its catching up. The bottom two timing diagrams in Fig.~\ref{mff}(c) refer to this scenario with $T_c =2 \Delta T$. As illustrated, the lag between the fusion filter and the localized filters grows exponentially with time in this scenario.
An improvement to the fusion filter is suggested in the top timing diagram of Fig.~\ref{mff}(c), where the fusion filter uses the most recent local filtering density of the localized filters. This allows the fusion filter to catch up with the localized filter even for cases $T_c > \Delta T$. Such a modified fusion implementation requires an updated fusion rule for the global posterior density, which is considered next.

At iteration $k+m$, we assume that node $l$, for ($1 \leq l \leq N$), has a particle-based approximation of the local filtering distributions $P(\x(k\!+\!m)|\z^{(l)}(1\!:\!k\!+\!m))$, while its fusion filter has a particle-based approximation of the global posterior distribution $P(\x(0\!:\!k)|\z(1\!:\!k))$ for iteration $k$.
In other words, the fusion filters are lagging the localized filters by $m$ iterations. In
the conventional fusion filter the statistics of $P(\x(k\!+\!1)|\z^{(l)}(1\!:\!k\!+\!1))$, for ($1 \leq l \leq N$) are used in the next consensus step of the fusion filter which then computes the global posterior $P(\x(0\!:\!k\!+\!1)|\z(1\!:\!k\!+\!1))$ based on Theorem~\ref{fprotocol}.
The modified~fusion filter uses the most recent local filtering distributions $P(\x(k\!+\!m)|\z^{(l)}(1\!:\!k\!+\!m))$ according to Theorem~\ref{fprotocol4}.
\begin{thm}\label{fprotocol4}
Conditioned on the state variables, assume that the observations
made at node $l$ are independent of the observations made at node $j$,
($j \neq l$). The global posterior distribution for a $N$--sensor
network at iteration $k\!+\!m$ is then given by
\begin{eqnarray} \label{ar-new1}
\lefteqn{ P\left(\x(0\!:\!k\!+\!m)|\z(1\!:\!k\!+\!m)\right) \propto \nonumber} \\
&&~~~\prod_{l=1}^N \frac{\prod_{k'=k+1}^{k+m} P
\left( \x(k')|\z^{(l)}(1\!:\!k')\right)}{\prod_{k'=k+1}^{k+m} P
\left( \x(k')|\z^{(l)}(1\!:\!k'\!-\!1)\right)}
\prod_{k'=k+1}^{k+m} P\left(\x(k')|\x(k'\!-\!1)\right)
 \times P\left(\x(0\!:\!k)|\z(1\!:\!k)\right).
\end{eqnarray}
\end{thm}
The proof of Theorem~\ref{fprotocol4} is included in Appendix~\ref{app:D}.
In the consensus step of the modified fusion filter, two average consensus algorithms are used to compute
$\prod_{l=1}^N \prod_{k'=k+1}^{k+m} P( \x(k')|\z^{(l)}(1\!:\!k'))$ and $\prod_{l=1}^N \prod_{k'=k+1}^{k+m} P( \x(k')|\z^{(l)}(1\!:\!k'\!-\!1))$, i.e.,
\begin{eqnarray} \label{mff:eq1}
\prod_{l=1}^N \prod_{k'=k+1}^{k+m} P\left( \x(k')|\z^{(l)}(1\!:\!k')\right) \propto
\prod_{l=1}^N \mathcal{N}\big(\bm{\mu}^{(l)}(k\!+\!1\!:\!k\!+\!m),\bm{P}^{(l)}(k\!+\!1\!:\!k\!+\!m)\big)
\end{eqnarray}
\begin{eqnarray}
\mbox{ and }
\prod_{l=1}^N \prod_{k'=k+1}^{k+m} P
\left( \x(k')|\z^{(l)}(1\!:\!k'\!-\!1)\right) \propto
\prod_{l=1}^N \mathcal{N}\big(\bm{\upsilon}^{(l)}(k\!+\!1\!:\!k\!+\!m),\bm{R}^{(l)}(k\!+\!1\!:\!k\!+\!m)\big),
\end{eqnarray}
instead of  computing $\prod_{l=1}^N P(\x(k)|\z^{(l)}(1\!:\!k))$ and $\prod_{l=1}^N P( \x(k)|\z^{(l)}(1\!:\!k\!-\!1))$ as was the case for the conventional fusion filter.
The modified fusion filter starts with a set of particles $ \cx_i^{(\MGF,l)}(k), W_i^{(\MGF,l)}(k)$ approximating $P(\x(0\!:\!k)|\z(1\!:\!k))$ and generates updated particles $\cx_i^{(\MGF,l)}(k\!+\!m), W_i^{(\MGF,l)}(k\!+\!m)$ for $P(\x(0\!:\!k\!+\!m)|\z(1\!:\!k\!+\!m))$  using the following weight update equation
\begin{eqnarray} \label{ar-new2}
W^{(l,\MGF)}_i(k\!+\!m)\!\propto\! W^{(l,\MGF)}_i(k) \!\times\!
\frac{\prod_{k'=k+1}^{k+m} P\left(\cx^{(l,\MGF)}_i(k')|\cx^{(l,\MGF)}_i(k'\!-\!1)\right) }
{\mathcal{N}\big(\cx^{(l,\MGF)}_i(k\!+\!m) ; \bm{\upsilon}(k\!+\!1\!:\!k\!+\!m),\bm{R}(k\!+\!1\!:\!k\!+\!m)\big)},
\end{eqnarray}
which is obtained directly from Eq.~(\ref{ar-new1}).
Note that the normal approximation in Eqs.~(\ref{mff:eq1})--(\ref{ar-new2}) are similar to the ones used in the conventional fusion filter.
Furthermore, we note that the modification requires prediction of the particles from iteration $k$ all the way to $k\!+\!m$ in order to evaluate the second term on the right hand side of Eq.~(\ref{ar-new2}). Algorithm~\ref{algo:Modifiedfusionfilter} outlines this step and summarizes the modified fusion filter.
\begin{algorithm*}[t!]
\caption{\textproc{~Modified Fusion Filter}
}
\label{algo:Modifiedfusionfilter}

\begin{algorithmic}[1]
\Input $\{\cx^{(l,\MGF)}_i(k), W^{(l,\MGF)}_i(k)\}_{i=1}^{N_{l,\MGF}}$ -
Fusion filter's particles and associated weights.
\Output $\{\cx^{(l,\MGF)}_i(k\!+\!m), W^{(l,\MGF)}_i(k\!+\!m)\}_{i=1}^{N_s}$ updated particles and associated weights.


\For{$k'=k\!+\!1:k+m,$}

$\mathcal{N}\big(\bm{\mu}^{(l)}(k'),\bm{P}^{(l)}(k') \big) =$ SaveGaussian$\left(\{ \cx_i^{(l)}(k'), W_i^{(l)}(k')\}_{i=1}^{N_s}\right)$

$\mathcal{N}\big(\bm{\upsilon}^{(l)}(k'),\bm{R}^{(l)}(k') \big) =$ SaveGaussian$\left(\{ \cx_i^{(l)}(k'\!+\!1|k'), W_i^{(l)}(k')\}_{i=1}^{N_s}\right)$ \EndFor

\State $\mathcal{N}\big(\bm{\mu}^{(l)}(k\!+\!1\!:\!k\!+\!m),\bm{P}^{(l)}(k\!+\!1\!:\!k\!+\!m)\big) =$ SaveGaussian$\left(  \prod_{k'=k\!+\!1}^{k+m} \mathcal{N}\big(\bm{\mu}^{(l)}(k'),\bm{P}^{(l)}(k') \big) \right)$.

\State $\mathcal{N}\big(\bm{\upsilon}^{(l)}(k\!+\!1\!:\!k\!+\!m),\bm{R}^{(l)}(k\!+\!1\!:\!k\!+\!m)\big) =$ SaveGaussian$\left(  \prod_{k'=k\!+\!1}^{k+m} \mathcal{N}\big(\bm{\upsilon}^{(l)}(k'),\bm{R}^{(l)}(k') \big) \right)$.


\State $\{\bm{\mu}^{(l,\MGF)}(k\!+\!1\!\!:\!k\!+\!m),
\bm{P}^{(l,\MGF)}(k\!+\!1\!\!:\!k\!+\!m)\} \!=\! $  DoFusion$\left(\{ \bm{\mu}^{(l)}(k\!+\!1\!\!:\!k\!+\!m),
\bm{P}^{(l)}(k\!+\!1\!\!:\!k\!+\!m)\}_{l=1}^{N}\right)$.

\State $\{\bm{\upsilon}^{(l,\MGF)}(k\!+\!1\!\!:\!k\!+\!m),
\bm{R}^{(l,\MGF)}(k\!+\!1\!\!:\!k\!+\!m)\} \!=\! $  DoFusion$\left(\{ \bm{\upsilon}^{(l)}(k\!+\!1\!\!:\!k\!+\!m),
\bm{R}^{(l)}(k\!+\!1\!\!:\!k\!+\!m)\}_{l=1}^{N}\right)$.


\For{$i=1:N_{\GF},$}

\For{$k'=k\!+\!1:k\!+\!m\!-\!1,$}

 $~~~~~~\cx^{(l,\MGF)}_i(k') \sim P\big(\x(k')|\cx_i^{(l,\MGF)}(k'\!-\!1)\big)$. \EndFor

$\cx^{(l,\MGF)}_i(k\!+\!m) \sim \mathcal{N}\big(\bm{\mu}^{(l,\MGF)}(k\!+\!1\!:\!k\!+\!m),\bm{P}^{(l,\MGF)}(k\!+\!1\!:\!k\!+\!m)\big)$.

 Compute weights $W_i^{(l,\MGF)}(k\!+\!m)$ using
Eq.~(\ref{ar-new2}).  \EndFor
\end{algorithmic}
\end{algorithm*}
%
\section{The Posterior Cram\'er-Rao Lower Bound} \label{sec:pcrlb}
Considering the non-linear filtering problem modeled in Eqs.~(\ref{sec.3:eq.1.1}) and~(\ref{sec.3:eq.1.2}) and the
  posterior distribution (Eq.~(\ref{chap3:sec2:eq8})) used in developing
the CF/DPF, the section computes the Posterior Cram\'er-Rao lower bound (PCRLB) for the distributed architecture.
We note that the PCRLB is independent of the estimation mechanism and
the bound should be the same for both centralized and distributed architectures. The question is whether the centralized expressions for computing the PCRLB are applicable to compute the PCRLB for other topologies, i.e., the hierarchical and distributed (decentralized)  architectures.
Reference~\cite{Tharmarasa:2009} considers a hierarchical architecture  with a central fusion center and shows that the centralized expressions can be used directly for the hierarchical case. The same authors argue in~\cite{Tharmarasa:2011} that the centralized expressions are no longer applicable for  distributed/decentralized architectures.
The exact expression for computing the
PCRLB for the distributed architecture is not yet available and only an approximate expression~\cite{Tharmarasa:2011} has recently been derived.
In this section, we derive the exact expression for computing the
PCRLB for the distributed topology. 
We note that our result is not restricted to the particle filter or the CF/DPF but is also applicable to any other distributed estimation approach.

The PCRLB inequality~\cite{Tichavsky:1998} states that the mean square
error (MSE) of the estimate $\hat{\x}(0\!:\!k)$ of the state variables
$\x(0\!:\!k)$ is lower bounded by
\begin{eqnarray}
\mathbb{E} \{(\hat{\x}(0\!:\!k)\!-\!\x(0\!:\!k))
(\hat{\x}(0\!:\!k)\!-\!\x(0\!:\!k))^T\} \!\!\geq\!\! [\bm{J}(\x(0\!:\!k))]^{-1},\!\!\!\!
\end{eqnarray}
where $\mathbb{E}$ is the expectation operator.  Matrix
$\bm{J}(\x(0\!:\!k))$ is referred to as the FIM~\cite{Tichavsky:1998}
derived from the joint probability density function (PDF)
$P(\x(0\!:\!k), \z(1\!:\!k))$.  Let $\nabla$ and $\Delta$,
respectively, be operators of the first and second order partial
derivatives given by $\nabla_{\x(k)} = \big[\frac{\partial}{\partial
    X_1(k)}, \ldots, \frac{\partial}{\partial X_{n_x}(k)}\big]^T $ and
$\Delta^{\x(k)}_{\x(k-1)} = \nabla_{\x(k-1)}\nabla_{\x(k)}^T$.  One
form of the Fisher information matrix $\bm{J}(\x(0\!:\!k))$ is defined
as~\cite{Tichavsky:1998}
\begin{eqnarray}
\bm{J}\big(\x(0\!:\!k)\big) = \mathbb{E} \big\{ -\Delta^{\x(0:k)}_{\x(0:k)} \log P(\x(0\!:\!k), \z(1\!:\!k))\big\}.
\end{eqnarray}
An alternative expression for the FIM can be derived from the
factorization $P(\x(0\!:\!k), \z(1\!:\!k)) = P(\x(0\!:\!k)|
\z(1\!:\!k))\times P(\z(1\!:\!k))$. Since $P(\z(1\!:\!k))$ is
independent of the state variables, we have
\begin{eqnarray} \label{ext.crlb.2}
\bm{J}\big(\x(0\!:\!k)\big) = \mathbb{E} \big\{ -\Delta^{\x(0:k)}_{\x(0:k)} \log P(\x(0\!:\!k)|\z(1\!:\!k))\big\}.
\end{eqnarray}
We now describe the centralized sequential formulation of the FIM.

\subsection{Centralized computation of the PCRLB}

Decomposing $\x(0\!:\!k)$ as $\x(0\!:\!k) = [\x^T(0\!:\!k\!-\!1),
  \x^T(k)]^T$ in $\bm{J}\big(\x(0\!:\!k)\big)$, Eq.~(\ref{ext.crlb.2})
simplifies to
\begin{eqnarray} \label{cpcrlb}
&&\!\!\!\!\!\!\!\!\!\!\!\!\!\!\!\bm{J}\big(\x(0\!:\!k)\big) \triangleq \left[
\begin{array}{cc}
\bm{A}^{11}(k) & \bm{A}^{12}(k)  \\
\bm{A}^{21}(k) & \bm{A}^{22}(k) \\
\end{array} \right] = \mathbb{E}\Bigg\{ -1\times
\left[
\begin{array}{ll}\label{ext.crlb.4}
\Delta^{\x(0:k-1)}_{\x(0:k-1)}  & \Delta^{\x(k)}_{\x(0:k-1)} \\
\Delta^{\x(0:k-1)}_{\x(k)}  & \Delta^{\x(k)}_{\x(k)} \\
\end{array} \right] \log P\big(\x(0\!:\!k)|\z(1\!:\!k)\big) \Bigg\}, \nonumber
\end{eqnarray}
provided that the expectations and derivatives exist. The information submatrix $\bm{J}(\x(k))$ for estimating $\x(k)$ is given by the inverse of the ($n_x \times
n_x$) right-lower block of $\big[\bm{J}\big(\x(0\!:\!k)\big)\big]^{-1}$. The information submatrix is computed using the matrix
inversion lemma~\cite{Tichavsky:1998} and given by
\begin{eqnarray} \label{ext.crlb.8}
\bm{J}(\x(k)) = \bm{A}^{22}(k)-\bm{A}^{21}(k)\big(\bm{A}^{11}(k)\big)^{-1}\bm{A}^{12}(k).
\end{eqnarray}
Proposition~\ref{pro1}~\cite{Tichavsky:1998} derives
$\bm{J}(\x(k))$  recursively  without manipulating the large ($kn_x\times kn_x$)
matrix $\bm{A}^{11}(k)$. The initial condition is given by
$\bm{J}\big(\x(0)\big) = \mathbb{E} \{ -\Delta^{\x(0)}_{\x(0)}
\log P(\x(0))\}$.

\begin{prop}\label{pro1}
The sequence $\{ \bm{J}\big(\x(k)\big)\}$ of local posterior
information sub-matrices for estimating state vectors $\x(k)$ at node
$l$, for ($1 \leq l \leq N$), obeys the following recursion
\begin{eqnarray} \label{ext.crlb.1}
&&\!\!\!\!\!\!\!\!\!\!\!\!\bm{J}\big(\x(k+1)\big) =  \bm{D}^{22}(k)
-\bm{D}^{21}(k)\Big( \bm{J}\big(\x(k)\big)+
\bm{D}^{11}(k)\Big)^{-1}\bm{D}^{12}(k)
\end{eqnarray}
%
\begin{eqnarray} \label{ext.crlb.10.1}
&&\!\!\!\!\!\!\!\!\!\!\!\!\!\!\!\!\!\!\!\!\!\mbox{where } \bm{D}^{11}(k) = \mathbb{E} \big\{ -\Delta^{\x(k)}_{\x(k)} \log P\big(\x(k+1)|\x(k)\big)\big\} \\
&&\!\!\!\!\!\!\!\!\!\!\!\!\!\!\!\!\!\!\!\!\!\bm{D}^{12}(k) = \big[ \!\bm{D}^{21}(k)\big]^T
\!\!\!\!= \mathbb{E}\big\{\!\!-\Delta^{\x(k+1)}_{\x(k)} \log P\big(\x(k+1)|\x(k)\big)\!\big\}~\\
&&\!\!\!\!\!\!\!\!\!\!\!\!\!\!\!\!\!\!\!\!\!\bm{D}^{22}(k) = \mathbb{E} \big\{ -\Delta^{\x(k+1)}_{\x(k+1)} \log P\big(\x(k+1)|\x(k)\big)\big\}+\underbrace{\mathbb{E} \big\{ -\Delta^{\x(k+1)}_{\x(k+1)} \log P\big(\z(k\!+\!1)|\x(k\!+\!1)\big)}_{\bm{J}(\z(k+1))}\big\}, \label{ext.crlb.10.2}
\end{eqnarray}
\end{prop}
The proof of Proposition~\ref{pro1} is given in~\cite{Tichavsky:1998}.
\noindent%
Conditioned on the state variables, the observations made at different nodes are independent
,therefore, $\bm{J}(\z(k+1))$~\cite{Hernandez:2002} in Eq.~(\ref{ext.crlb.10.2}) is simplifies to
\begin{eqnarray}
\label{xxxx}
\bm{J}(\z(k+1)) &\!\!\!=\!\!\!& \sum_{l=1}^{N} \bm{J}(\z^{(l)}(k+1)) = \sum_{l=1}^{N} \mathbb{E} \big\{ -\Delta^{\x(k+1)}_{\x(k+1)} \log P\big(\z^{(l)}(k\!+\!1)|\x(k\!+\!1)\big)\}. \nonumber
\end{eqnarray}
In other words, the expression for $\bm{J}(\x(k+1))$
(Eq.~(\ref{ext.crlb.1})) requires distributed information (sensor
measurement) only for computing $\bm{J}(\z(k+1))$. Other terms in
Eq.~(\ref{ext.crlb.1}) can be computed locally.  In the next section,
we derive the distributed PCRLB.

\vspace{.05in}
\noindent
\textit{B. Distributed computation of the PCRLB}

In the sequel, $\bm{J}^{(l)}\big(\x(0\!:\!k)\big)$, for ($1\leq l \leq
N$), denotes the local FIM corresponding to the local estimate of
$\x(0\!:\!k)$ derived from the local posterior density $P(\x(0\!:\!k)|\z^{(l)}(1\!:\!k))$.  Similarly,
$\bm{J}^{(l)}\big(\x(0\!:\!k\!+\!1|k)\big)$ denotes the local FIM
corresponding to the local prediction estimate of $\x(0\!:\!k\!+\!1)$
derived from the local prediction density $P(\x(0\!:\!k\!+\!1),
\z^{(l)}(1\!:\!k))$. The expressions for
$\bm{J}^{(l)}\big(\x(0\!:\!k)\big)$ and
$\bm{J}^{(l)}\big(\x(0\!:\!k\!+\!1|k)\big)$ are similar in nature to
Eq.~(\ref{cpcrlb}) except that the posterior density
$P(\x(0\!:\!k)|\z(1\!:\!k))$ is replaced by their corresponding local
posteriors.
The local FIM $\bm{J}^{(l)}\big(\x(k)\big)$ is given by the
inverse of the ($n_x \times n_x$) right-lower block of
$\big[\bm{J}^{(l)}\big(\x(0\!:\!k)\big)\big]^{-1}$. Similarly, the
prediction FIM $\bm{J}^{(l)}\big(\x(k\!+\!1|k)\big)$ is given by the
inverse of the ($n_x \times n_x$) right-lower block of
$\big[\bm{J}^{(l)}\big(\x(0\!:\!k\!+\!1|k)\big)\big]^{-1}$.

The problem we wish to solve is to compute the global information
sub-matrix, denoted by $\bm{J}\big(\x(k\!+\!1)\big)$, as a function of
the local FIMs $\bm{J}^{(l)}\big(\x(k\!+\!1)\big)$ and local
prediction FIMs $\bm{J}^{(l)}\big(\x(k\!+\!1|k)\big)$, for ($1 \leq l
\leq N$). Note that $\bm{J}^{(l)}\big(\x(k)\big)$ can be updated
sequentially using Eqs.~(\ref{ext.crlb.1})-(\ref{ext.crlb.10.2}) where
$\bm{J}(\z^{(l)}(k+1))$ replaces $\bm{J}(\z(k+1))$ in
Eq.~(\ref{ext.crlb.10.2}). Proposition~\ref{pro2} derives a recursive formula for
computing $\bm{J}^{(l)}\big(\x(k\!+\!1|k)\big)$, i.e., the FIM for the
local prediction distribution.
\begin{prop}\label{pro2}
The sequence $\{ \bm{J}^{(l)}\big(\x(k\!+\!1|k)\big) \}$ of the local
prediction information sub-matrices for predicting state vectors
$\x(k)$ at node $l$, for ($1 \leq l \leq N$), follows the
recursion
\begin{eqnarray}\label{ext.crlb.6}
\bm{J}^{(l)}\big(\x(k\!+\!1|k)\big) =  \bm{B}^{22}(k)-\bm{D}^{21}(k)\big( \bm{J}^{(l)}\big(\x(k)\big)+
\bm{D}^{11}(k)\big)^{-1}\bm{D}^{12}(k)
\end{eqnarray}
where $\bm{J}^{(l)}\big(\x(k)\big)$ is given by
Eq.~(\ref{ext.crlb.1}), $D^{11}(k)$, $D^{12}(k)$, and $D^{21}(k)$ are
given by Eqs.~(\ref{ext.crlb.10.1})-(\ref{ext.crlb.10.2}) and
\begin{eqnarray}\label{ext.crlb.7-1}
\bm{B}^{22}(k) = \mathbb{E} \big\{ -\Delta^{\x(k+1)}_{\x(k+1)} \log P\big(\x(k+1)|\x(k)\big)\big\}, \label{ext.crlb.7-3}
\end{eqnarray}
\end{prop}
The proof of Proposition~\ref{pro2} is included in Appendix~\ref{app:B}.
Theorem~\ref{pro3} is our main result. It provides the
exact recursive formula for
computing the distributed FIM corresponding to the global estimation from the
local FIMs $\bm{J}^{(l)}(\x(k))$ and local prediction FIMs $\bm{J}^{(l)}(\x(k\!+\!1))$.
\begin{thm}\label{pro3}
The sequence $\{ \bm{J}\big(\x(k)\big) \}$ of information sub-matrices
corresponding to global estimates follows the  recursion
\begin{eqnarray}\label{ext.crlb.13}
\bm{J}\big(\x(k\!+\!1)\big) = \bm{C}^{22}(k)-
\bm{D}^{21}(k)\big( \bm{J}\big(\x(k)\big)+\bm{D}^{11}(k)\big)^{-1}\bm{D}^{12}(k)
\end{eqnarray}
where $\bm{D}^{11}(k)$, $\bm{D}^{21}(k)$, and $\bm{D}^{12}(k)$ are
given by Eqs.~(\ref{ext.crlb.10.1})-(\ref{ext.crlb.10.2}) and
\begin{eqnarray}\label{ext.crlb.14.1}\label{ext.crlb.14.3}
\bm{C}^{22}(k) \!=\! \sum_{l=1}^N \bm{J}^{(l)}(\x(k\!+\!1)) - \sum_{l=1}^N \bm{J}^{(l)}(\x(k\!+\!1|k))
+ \mathbb{E} \big\{ \!\!-\!\!\Delta^{\x(k+1)}_{\x(k+1)} \log P\big(\x(k+1)|\x(k)\big)\big\},
\end{eqnarray}
where $\bm{J}^{(l)}(\x(k\!+\!1))$ and $\bm{J}^{(l)}(\x(k\!+\!1|k))$
are defined in Prepositions~\ref{pro1} and~\ref{pro2}, respectively.
\end{thm}
The proof of Theorem~\ref{pro3} is included in Appendix~\ref{app:C}.
In~\cite{Tharmarasa:2011}, an approximate updating equation based~on
the information filter (an alternative form of the Kalman filter) is proposed for computing $\bm{J}(\x(k\!+\!1))$ at node $l$ which 
is represented in our notation as follows
\begin{eqnarray}\label{extention.2}
\hat{\bm{J}}(\x(k\!+\!1)) \!=\! \bm{J}^{(l)}(\x(k\!+\!1)) + \sum_{j \neq l} \bigg(\bm{J}^{(j)}(\x(k\!+\!1))\!-\! \bm{J}^{(j)}(\x(k\!+\!1|k))\bigg).
\end{eqnarray}
Term $\bm{J}^{(l)}(\x(k\!+\!1))$ is given by
\begin{eqnarray} \label{extention.3}
&&\!\!\!\!\!\!\!\!\!\!\!\!\bm{J}^{(l)}\big(\x(k+1)\big) =  \bm{D}^{22}(k)
-\bm{D}^{21}(k)\Big( \bm{J}^{(l)}\big(\x(k)\big)+
\bm{D}^{11}(k)\Big)^{-1}\bm{D}^{12}(k).
\end{eqnarray}
%
Subtracting Eq.~(46) from the above equation and rearranging, we get
\begin{eqnarray}\label{extention.3.2}
 \bm{D}^{22}(k)= \bigg(\bm{J}^{(l)}(\x(k\!+\!1))\!-\! \bm{J}^{(l)}(\x(k\!+\!1|k))\bigg) + \bm{B}^{22}(k).
\end{eqnarray}
Substituting $\bm{B}^{22}(k)$ from Eq.~(47) in Eq.~\eqref{extention.3.2} and then substituting the resulting expression of $ \bm{D}^{22}(k)$ in Eq.~\eqref{extention.3}, we get
\begin{eqnarray}\label{extention.4}
\lefteqn{\!\!\!\!\!\!\!\!\!\!\!\!\!\!\!\!\!\!\!\!\!\!\!\!\!\!\!\!\!\!\!\!\!\!\!\!\!\!\!\!\!\!\!\!\!\!\!\!\!\!\!\!\!\!\!\!\bm{J}^{(l)}\big(\x(k\!+\!1)\big) \!\!=\!\!  \bigg(\bm{J}^{(l)}(\x(k\!+\!1))\!-\! \bm{J}^{(l)}(\x(k\!+\!1|k))\bigg) \!\!+\!\!  \mathbb{E} \big\{ -\Delta^{\x(k+1)}_{\x(k+1)} \log P\big(\x(k+1)|\x(k)\big)\big\} \nonumber}\\
&&-\bm{D}^{21}(k)\big( \bm{J}^{(l)}\big(\x(k)\big)+
\bm{D}^{11}(k)\big)^{-1}\bm{D}^{12}(k).
\end{eqnarray}
Substituting Eq~(\ref{extention.4}) in Eq.~(\ref{extention.2}) results in the following approximated fused FIM at node $l$
\begin{eqnarray}\label{extention.5}
\hat{\bm{J}}\big(\x(k\!+\!1)\big) = \bm{C}^{22}(k)-
\bm{D}^{21}(k)\big( \bm{J}^{(l)}\big(\x(k)\big)+\bm{D}^{11}(k)\big)^{-1}\bm{D}^{12}(k).
\end{eqnarray}
Note the differences between Eqs.~(\ref{extention.5}) and~\eqref{ext.crlb.13}. The second term in the right hand side of Eq.~\eqref{extention.5}
is based on the previous local FIM $\bm{J}^{(l)}(\x(k))$ at node $l$ thus making it node-dependent,
while the corresponding term in Eq.~\eqref{ext.crlb.13} is based on the overall FIM from the previous iteration.
When the PCRLB is computed in a distributed manner, Eq.~\eqref{extention.5} differs from one node to another and is, therefore,
not conducive for deriving the overall PCRLB through consensus. To make Eq.~\eqref{extention.2} node independent,
our simulations also compare the proposed exact PCRLB with another approximate expression, which only includes the first two terms of $\bm{C}^{22}(k)$ in Eq.~(\ref{ext.crlb.14.1})~i.e.,
\begin{eqnarray}\label{app.1}
\hat{\bm{J}}(\x(k\!+\!1)) \cong \bm{C}^{22}(k) \cong \sum_{l=1}^N \bigg(\bm{J}^{(l)}(\x(k\!+\!1))\!-\! \bm{J}^{(l)}(\x(k\!+\!1|k))\bigg).
\end{eqnarray}
The expectation terms in
Eqs.~(\ref{ext.crlb.10.1})-(\ref{ext.crlb.10.2}),~(\ref{ext.crlb.7-1}),
and~(\ref{ext.crlb.14.1}) can be further simplified for additive
Gaussian noise, i.e., when $\bm{\xi}(\cdot)$ and
$\bm{\zeta}^{(l)}(\cdot)$ are normally distributed with zero mean and
covariance matrices $\bm{Q}(k)$ and $\bm{R}^{(l)}(k)$, respectively.
In this case, we have
\begin{eqnarray} \label{ext.crlb.15.1}
\!\!\!\!\!\!\!\!\!\!\!\!\!\!\!\!\!\!&& \bm{D}^{11}(k) \!=\! \mathbb{E} \big\{ \big[\nabla_{\x(k)} \bm{f}^T(k) \big] \bm{Q}^{-1}(k) [\nabla_{\x(k)} \bm{f}^T(k) \big]^T \big\}\\
\!\!\!\!\!\!\!\!\!\!\!\!\!\!\!\!\!\!&& \bm{D}^{12}(k) \!=\! - \mathbb{E} \big\{ \big[\nabla_{\x(k)} \bm{f}^T(k) \big] \bm{Q}^{-1}(k) = \big[\bm{D}^{21}(k)\big]^T \\
\!\!\!\!\!\!\!\!\!\!\!\!\!\!\!\!\!\!&& \bm{D}^{22}(k) \!=\! \bm{Q}^{-1}(k) + \mathbb{E} \big\{ \big[\nabla_{\x(k+1)} \bm{g}^{(l)^T}(k\!+\!1) \big] \bm{R}^{(l)^{-1}}(k\!+\!1) [\nabla_{\x(k+1)} \bm{g}^{(l)^T}(k\!+\!1) \big]^T \big\}\\
\!\!\!\!\!\!\!\!\!\!\!\!\!\!\!\!\!\!&& \bm{B}^{22}(k) \!=\! \bm{Q}^{-1}(k),
\end{eqnarray}
\begin{eqnarray}
\mbox{and} \quad \quad\!\!\!\!\!\!\!\!\!\!\!\!\!\!\!\!\!\!&& \bm{C}^{22}(k) = \sum_{l=1}^N \!\!\bm{J}^{(l)}(\x(k\!+\!1))\!-\!\! \sum_{l=1}^N \!\!\bm{J}^{(l)}(\x(k\!+\!1|k)) + \bm{Q}^{-1}(k). \label{ext.crlb.15.2}~~~
\end{eqnarray}
Note that Theorem~\ref{pro3}
(Eqs.~(\ref{ext.crlb.13})-(\ref{ext.crlb.14.3})) provides a recursive
framework for computing the distributed PCRLB.  Further, the
proposed distributed PCRLB can be implemented in a distributed
fashion because Eq.~(\ref{ext.crlb.14.3}) has only two summation terms
involving local parameters. These terms can be computed in a
distributed manner using the average consensus
algorithms~\cite{Arash:ssp2011}. Other terms in
Eqs.~(\ref{ext.crlb.13})-(\ref{ext.crlb.14.3}) are dependent only on
the process model and can be derived locally.  In cases
(non-linear/non-Gaussian dynamic systems) where direct computation of
$\bm{D}^{11}(k)$, $\bm{D}^{12}(k)$, $\bm{D}^{21}(k)$,
$\bm{D}^{22}(k)$, $\bm{B}^{22}(k)$, and $\bm{C}^{22}(k)$ involves
high-dimensional integrations, particle filters can alternatively be used to compute these terms.
\section{Simulations: Bearing Only Target Tracking} \label{sec:simu}
A distributed bearing-only tracking (BOT)
application~\cite{Bearing:gordon} is simulated to test the proposed
CF/DPF. The BOT problem arises in a variety of nonlinear signal
processing applications including radar surveillance, underwater
submarine tracking in sonar, and robotics~\cite{Bearing:gordon}.
 The objective is to design a practical filter capable of estimating the state kinematics $\x(k) = [X(k), Y(k), \dot{X}(k), \dot{Y}(k)]$
(position $[X, Y]$ and velocity $[\dot{X} , \dot{Y}]$) of the target
from the bearing angle measurements and prior knowledge of the
target's motion.  The BOT is inherently a nonlinear application with
nonlinearity incorporated either in the state dynamics or in the
measurement model depending on the choice of the coordinate system
used to formulate the problem.
In this paper, we consider non-linear target kinematics with a non-Gaussian observation model.
A clockwise coordinated turn kinematic motion model~\cite{Bearing:gordon}  given by
\begin{eqnarray}
\bm{f}(\x(k)) = \left[
\begin{array}{cccc}
1 & 0 & \frac{\sin(\Omega(k)\Delta T)}{\Omega(k)} & -\frac{1-\cos(\Omega(k)\Delta T)}{\Omega(k} \\
0 & 1 & \frac{1-\cos(\Omega(k)\Delta T)}{\Omega(k)} & \frac{\sin(\Omega(k)\Delta T)}{\Omega(k)} \\
0 & 0 & \cos(\Omega(k)\Delta T)& -\sin(\Omega(k)\Delta T) \\
0 & 0 & \sin(\Omega(k)\Delta T)& \cos(\Omega(k)\Delta T) \\
\end{array} \right], \quad \mbox{with}\quad
\Omega(k) =
\frac{A_m}{\sqrt{(\dot{X}(k))^2+(\dot{Y}(k))^2}},
\end{eqnarray}
is considered with the manoeuvre acceleration parameter $A_m$ set to  $1.08\times10^{-5}\mbox{units/s}^2$~\cite{Bearing:gordon}.
A sensor network of $N = 20$ nodes with random geometric
graph model in a square region of dimension ($16 \times 16$) units is considered. Each sensor communicates only with its neighboring nodes within a connectivity radius of $\sqrt{2\log(N)/N}$ units.
In addition, the network is assumed to be connected with each node linked to  at least one other node in the network. Measurements are the target's bearings with respect to
the platform of
 each node (referenced clockwise positive to the $y$-axis), i.e.,
\begin{eqnarray}
  Z^{(l)}(k) &=& {\mbox{atan}\left(\frac{X(k)-X^{(l)}}{Y(k)-Y^{(l)}}\right)} + \zeta^{(l)}(k),
\end{eqnarray}
where $(X^{(l)}, Y^{(l)})$ are the coordinates of node $l$.
The observations are assumed to be corrupted by the non-Gaussian \textit{target glint noise}~\cite{Rong:2005} modeled as a
 mixture model of two zero-mean Gaussians~\cite{Rong:2005},
one with a high probability of occurrence and small variance
and  the other with relatively a small probability of occurrence and high variance. The likelihood model at node $l$, for ($1 \leq l \leq N$), is described as
\begin{eqnarray}
P(\z^{(l)}|\x(k)) =  (1-\epsilon) \times \mathcal{N}(\x; 0,\sigma^{2}_{\zeta^{(l)}}(k)) +\epsilon \times \mathcal{N}(\x; 0,10^4\sigma^{2}_{\zeta^{(l)}}(k)),
\end{eqnarray}
where $\epsilon = 0.09$ in the simulations.
Furthermore, the observation noise is assumed to be state dependent
such that the bearing noise variance $\sigma^{2}_{\zeta^{(l)}}(k)$ at
node $l$ depends on the distance $r^{(l)}(k)$ between the observer and
target. Based on~\cite{Chung:2005}, the variance of the observation noise at node $l$ is, therefore, given by
\begin{eqnarray}
\sigma^{2}_{\zeta^{(l)}}(k) = 0.08r^{(l)^2}(k) + 0.1150r^{(l)}(k) + 0.7405.
\end{eqnarray}
Due to state-dependent noise variance, we  note that the signal to noise ratio (SNR) is time-varying and differs (within a range of $-10$dB to $20$dB) from one sensor node to the other depending on the location of the target.
Averaged across all nodes and time, the mean SNR is $5.5$dB.
In our simulations, we chose to incorporate observations made at all nodes in the estimation, however, sensor selection based on the proposed distributed PCRLB can be used, instead, which will be considered~as future work.
Both centralized and distributed filters are initialized based on the procedure described~in~\cite{Bearing:gordon}.

\noindent
\textbf{Simulation Results:} The target starts from coordinates $(3, 6)$ units
The position of target the target ($[X, Y]$) in first three iterations are $(2.6904, 5.6209)$,  $(2.3932, 5.2321)$, and $(2.1098, 4.8318)$. The initial course is set at $-140^\circ$ with the standard deviation of the process noise
$\sigma_v = 1.6 \times 10^{-3}$ unit.
The number $N_s$ of vector particles for centralized implementation is $N_s=10,000$. The number $N_{\text{LF}}$ and $N_{\text{FF}}$ of vector particles used in each local filter and fusion filter is $500$. The number of particles for the CF/DPF  are selected to keep its computational complexity the same as that of the centralized implementation.
To quantify the tracking performance of the proposed methods two scenarios are considered as follows.

\noindent
\textbf{Scenario 1:}
accomplishes three goals. First, we compare the performance of the proposed CF/DPF
versus the centralized implementation. The fusion filters used in the CF/DPF are allowed to converge between two consecutive iterations of the localized particle filters (i.e., we follow the timing subplot (a) of Fig.~\ref{mff}).
Second, we compare the impact of the three proposal distributions listed in Section~III-E
on the CF/DPF. The performance of the CF/DPF is computed for each of these proposal distributions using Monte Carlo simulations.
Third, we plot the proposed PCRLB computed from the distributed configuration and compare it with its counterpart obtained from the  centralized architecture that includes a fusion center.

Fig.~\ref{Target_track}(a) plots one realization of the target
track and the estimated tracks obtained from: (i) The CF/DPF; (ii) the centralized implementation, and; (iii) a single node estimation (stand alone case ).
In the CF/DPF, we used the Gaussian approximation of the optimal proposal distribution as the proposal distribution (Case $3$ in Section~III-E). 
The two estimates from the CF/DPF and the centralized implementation are fairly close to the
true trajectory of the target so much so as that they overlap.
The stand alone scenario based on running a particle filter at a single node (shown as the red circle in Fig.~\ref{Target_track}(a)) fails to track the target.
Fig.~\ref{Target_track}(b) plots the cumulative distribution function (CDF)
for the X-coordinate of the target estimated using the centralized and CF/DPF implementations for iterations $k = 5$ and $22$.
We note that the two CDFs are close to each other.
Fig.~\ref{Target_track} illustrates the near-optimal nature of the CF/DPF.
\begin{figure}[!t]
\mbox{\subfigure[]{\includegraphics[scale=0.58]{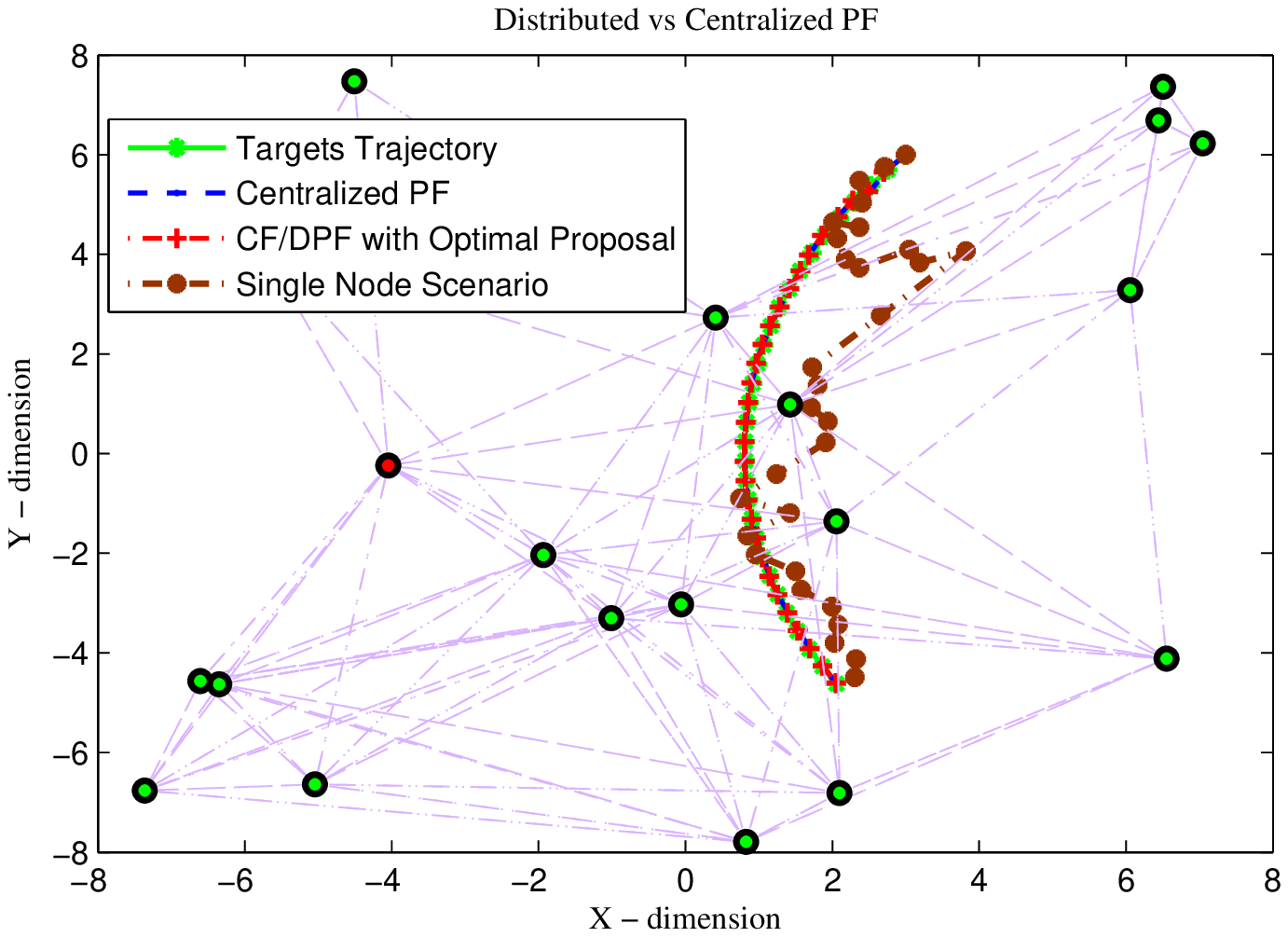}}
\subfigure[]{ \includegraphics[scale=0.55]{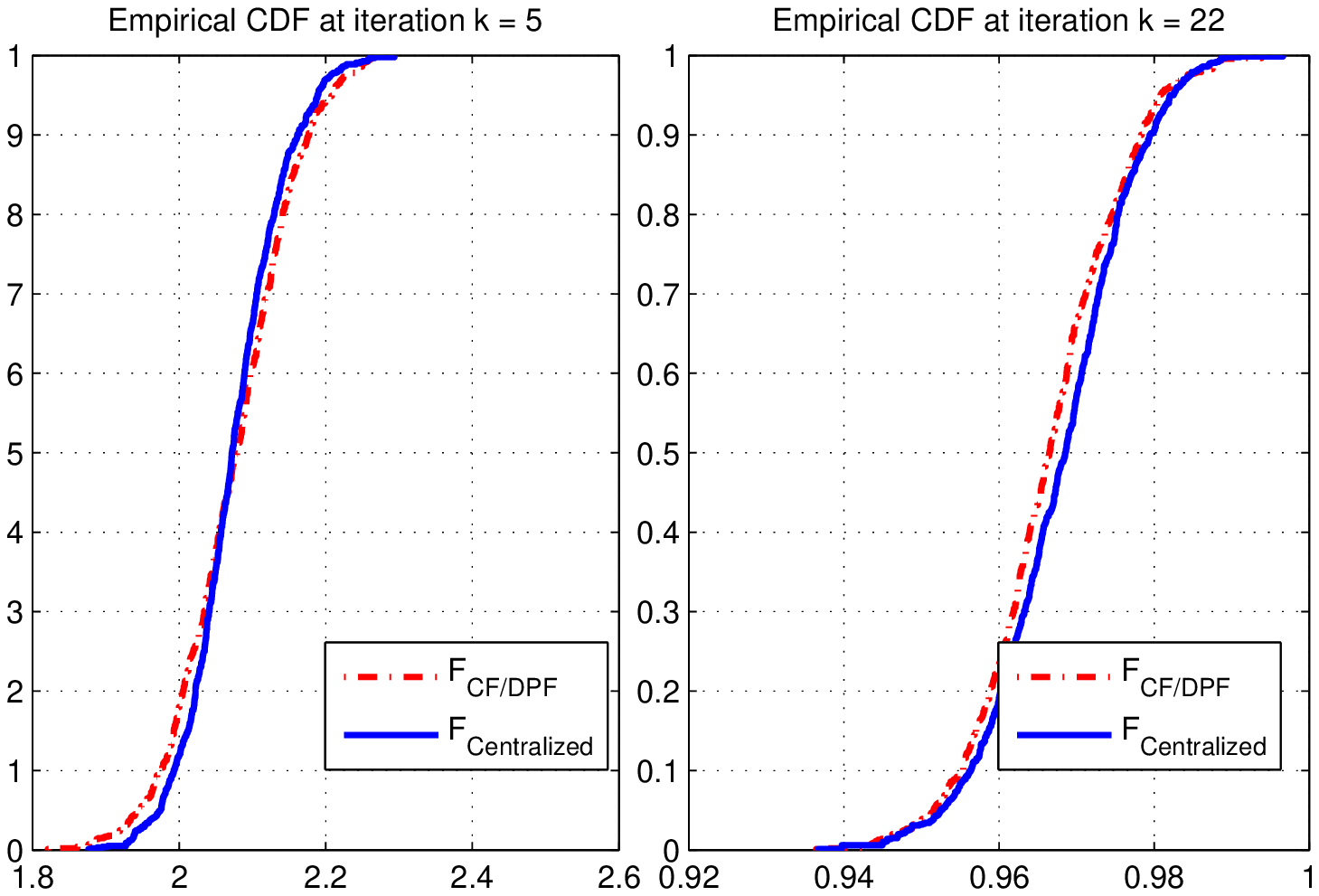}}}
\caption{\label{Target_track} Scenario 1:  (a) Actual target's track obtained from the centralized,
CF/DPF and stand-alone algorithms. Here, the consensus algorithm is allowed
to converge. (b) CDFs for the X-coordinate of the target estimated using the centralized
and CF/DPF for $k = 5, 22$.}
\end{figure}

Fig.~\ref{Target_RMS} compares the root mean square (RMS) error curves for the target's position.
Based on a Monte-Carlo simulation of $100$ runs,
Fig.~\ref{Target_RMS}(a) plots the RMS error curves for the estimated target's position via three
CF/DPF implementations obtained using different proposals distributions.
We observe that the SIR fusion filter performs the worst in this highly non-linear environment with non-Gaussian observation noise,  while the outputs of the centralized and the other two distributed implementations are fairly close to each other and approach the PCRLB.
Since the product fusion filter requires less computations, the simulations in Scenario~$2$ are based on the CF/DPF implementation using the product fusion filter.
Second, in Fig.~\ref{Target_RMS}(b), we compare the PCRLB obtained from the distributed and centralized architectures.
The Jacobian terms $\nabla_{\x(k)} \bm{f}^T(k)$ and $\nabla_{\x(k+1)}
\bm{g}^{(l)^T}(k\!+\!1)$, needed for the PCRLB, are computed using the procedure outlined in~\cite{Bearing:gordon}.
  It is observed that the bound obtained from the proposed
distributed computation of the PCRLB is more accurate and closer to the PCRLB computed from the centralized expressions.
As expected, the proposed decentralized PCRLB overlaps with its centralized
counterpart.
\begin{figure}[!t]
\centering
\mbox{\subfigure[]{ \includegraphics[scale=0.5]{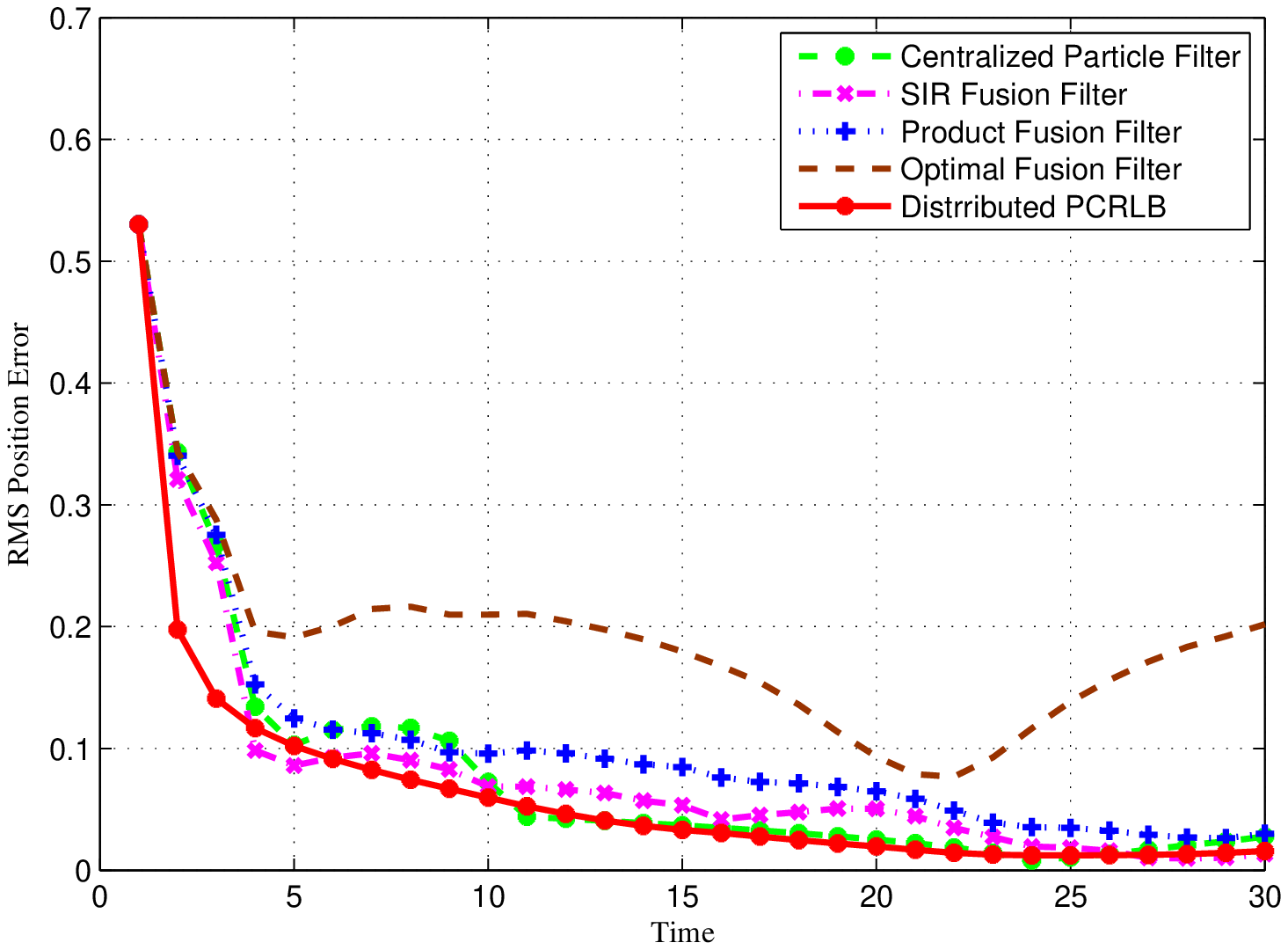}} \quad \subfigure[]{\includegraphics[scale=0.5]{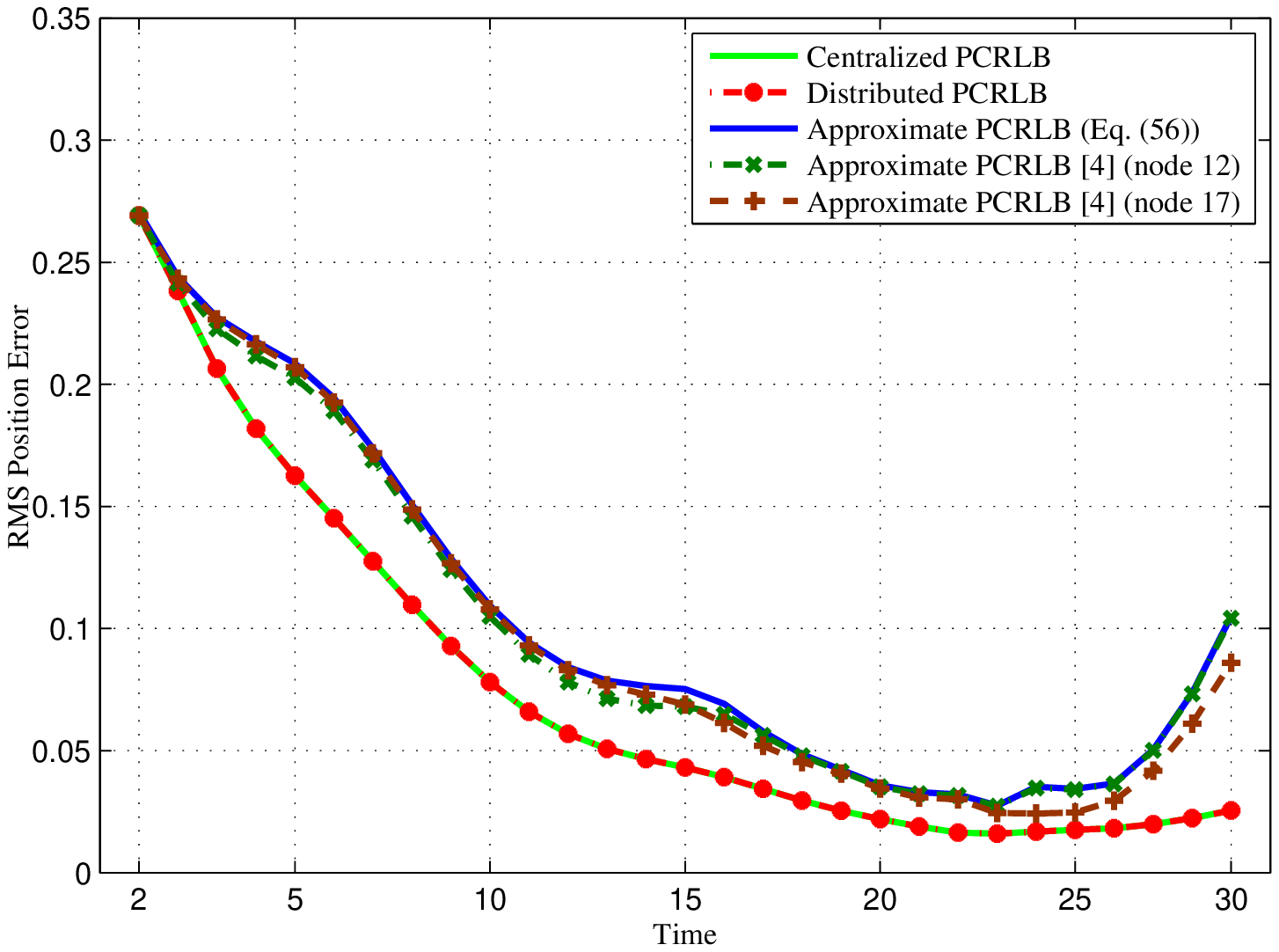}}}
\caption{\label{Target_RMS} Scenario 1:
 (a) Comparison of the RMS errors resulting from the centralized
  versus distributed implementations.  (b) Comparison of the proposed
  exact PCRLB, its approximated value based on~\cite{Tharmarasa:2011}
  (Eq.~\eqref{extention.5}) computed at node $12$ and $17$, and
  approximation of Eq.~\eqref{app.1}.  The PCRLBs computed using
  the centralized and distributed (proposed) expressions overlap so
  that they are virtually indistinguishable.
}
\end{figure}

\noindent
\textbf{Scenario 2:} The second scenario models the timing subplot (c)
of Fig.~\ref{mff}.  The convergence of the fusion filter takes up to
two iterations of the localized filters.  The original fusion filter
(Algorithm~1) is unable to converge within two consecutive iterations
of the localized particle filters.  Therefore, the lag between the
fusion filters and the localized filters in the CF/DPF continues to
increase exponentially.  The modified fusion filter described in
Algorithm~\ref{algo:Modifiedfusionfilter} is implemented to limit the
lag to two localized filter iterations.
The target's track are shown in Fig.~\ref{MFF_TT_nl}(a) for the
centralized implementation, original and modified fusion filter.
Fig.~\ref{MFF_TT_nl}(b) shows the RMS error curves for the target's
position including the RMS error resulting from Algorithm~$1$. Since consensus
is not reached in Algorithm~1, therefore, the fusion estimate is
different from one node to another. For Algorithm~$1$,
Fig.~\ref{MFF_TT_nl}(b) includes the RM error associated with a
randomly selected node. In Fig.~\ref{MFF_TT_nl}(b), Algorithm $1$
performs poorly due to consensus not reached, while the modified
fusion filter performs reasonably well. The performance of the
modified fusion filter remains close to its centralized counterpart,
therefore, it is capable of handling intermittent consensus~steps.


\noindent
\textbf{Scenario 3:}
%
\begin{figure}[!t]
\centering
\mbox{\subfigure[]{\includegraphics[scale=0.5]{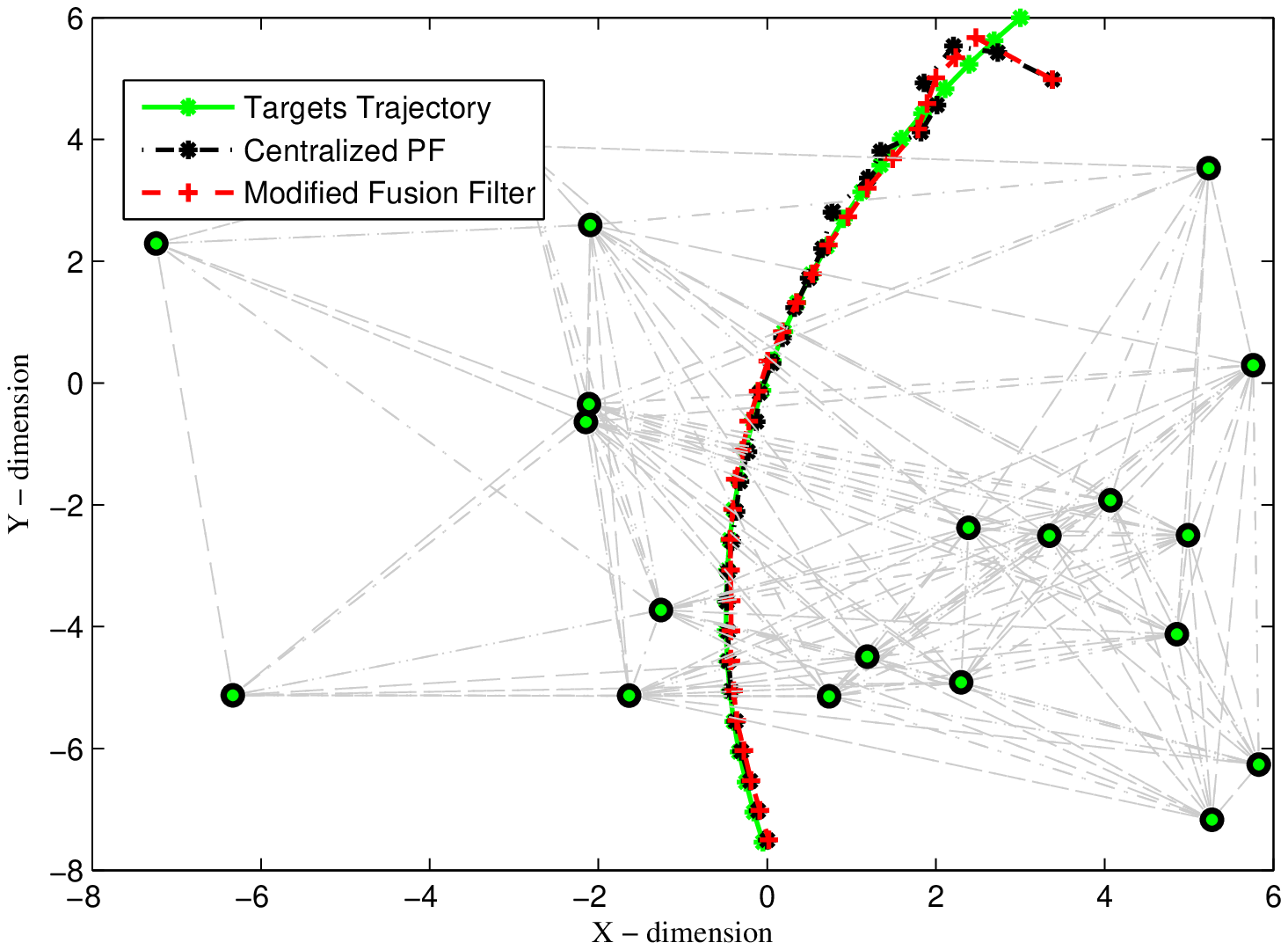}}\quad
\subfigure[]{ \includegraphics[scale=0.5]{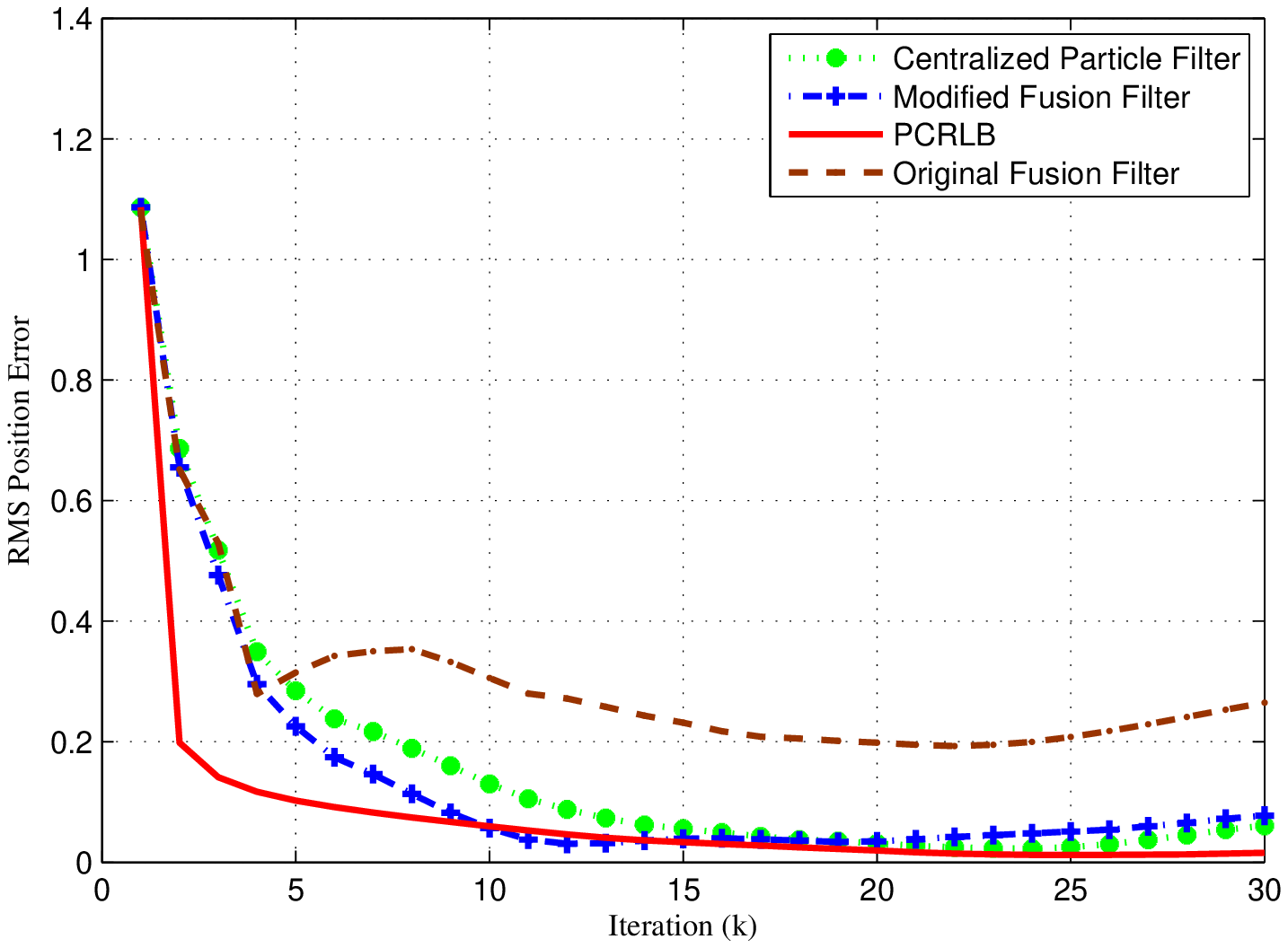}}}
\caption{\label{MFF_TT_nl} Scenario 2:  (a) Actual target's track alongside estimated tracks obtained from the centralized and modified fusion filter. Here, the consensus algorithm converges after each $2$ iterations of the local particle filters. (b) Comparison of the RMS errors resulting from the centralized, original fusion filter and  modified fusion filter.}
\end{figure}
In the third scenario, we consider a distributed mobile robot
localization problem~\cite{Simonetto:2010, Thurn:2005}
based on angle-only measurements. This is a good benchmark since the
underlying dynamics is non-linear with non-additive forcing terms resulting in a non-Gaussian transitional state model.
The state vector of the unicycle
robot is defined by $\x =[X, Y, \theta]^T$, where ($X, Y$) is the $2$D
coordinate of the robot and $\theta$ is its orientation.  The
velocity and angular velocity are denoted by $\tilde{V}(k)$ and $\tilde{W}(k)$,
respectively.  The following discrete-time non-linear unicycle model~\cite{Thurn:2005}
represents the state dynamics of the robot
\begin{eqnarray}\label{rob1}
X(k\!+\!1) &=& X(k)+\frac{\tilde{V}(k)}{\tilde{W}(k)}\left( \sin\Big( \theta(k)+\tilde{W}(k)\Delta T\Big) - \sin\big( \theta(k)\big)\right),  \\
Y(k\!+\!1) &=& Y(k)+\frac{\tilde{V}(k)}{\tilde{W}(k)}\left( \cos\Big( \theta(k)+\tilde{W}(k)\Delta T\Big) - \cos\big( \theta(k)\big)\right), \label{rob2}\\
\text{and }\theta(k\!+\!1) &=& \theta(k)+\tilde{W}(k)\Delta T+\xi_{\theta} \Delta T,
\end{eqnarray}
where $\Delta T$ is the sampling time and $\xi_{\theta}$ is the orientation
noise term.
The design parameters are: $\Delta T =1$, a mean velocity of $30$
cm/s with a standard deviation of $5$ cm/s, and a
mean angular velocity of $0.08$ rad/s with a standard deviation of $0.01$ rad.
Because of the presence of sine and cosine functions, the overall
state dynamics in Eq.~\eqref{rob1}-\eqref{rob2} are in effect
perturbed by non-Gaussian terms.  The observation model is similar to
the one described for Scenario~$1$ with non-Gaussian and
state-dependent observation noise.  The robot starts at coordinates
$(3,5)$. Fig.~\ref{robot} (a) shows one realization of the sensor
placement along with the estimated robot's trajectories obtained from
the proposed CF/DPF, centralized particle filter and distributed
unscented Kalman filter (UKF)~\cite{Simonetto:2010} implementations.
We observe that both centralized particle filter and CF/DPF clearly
follow the robot trajectory, while the distributed UKF deviates after
a few initial iterations.  Fig.~\ref{robot} (b) plots the RMS errors
obtained form a Monte-Carlo simulation of $100$ runs, which
corroborate our earlier observation that the CF/DPF and the
centralized particle filter provide better estimates that are close to
each other. The UKF produces a different result
with the highest error.
\begin{figure}[t]
\centering
\mbox{\subfigure[]{\includegraphics[scale=0.5]{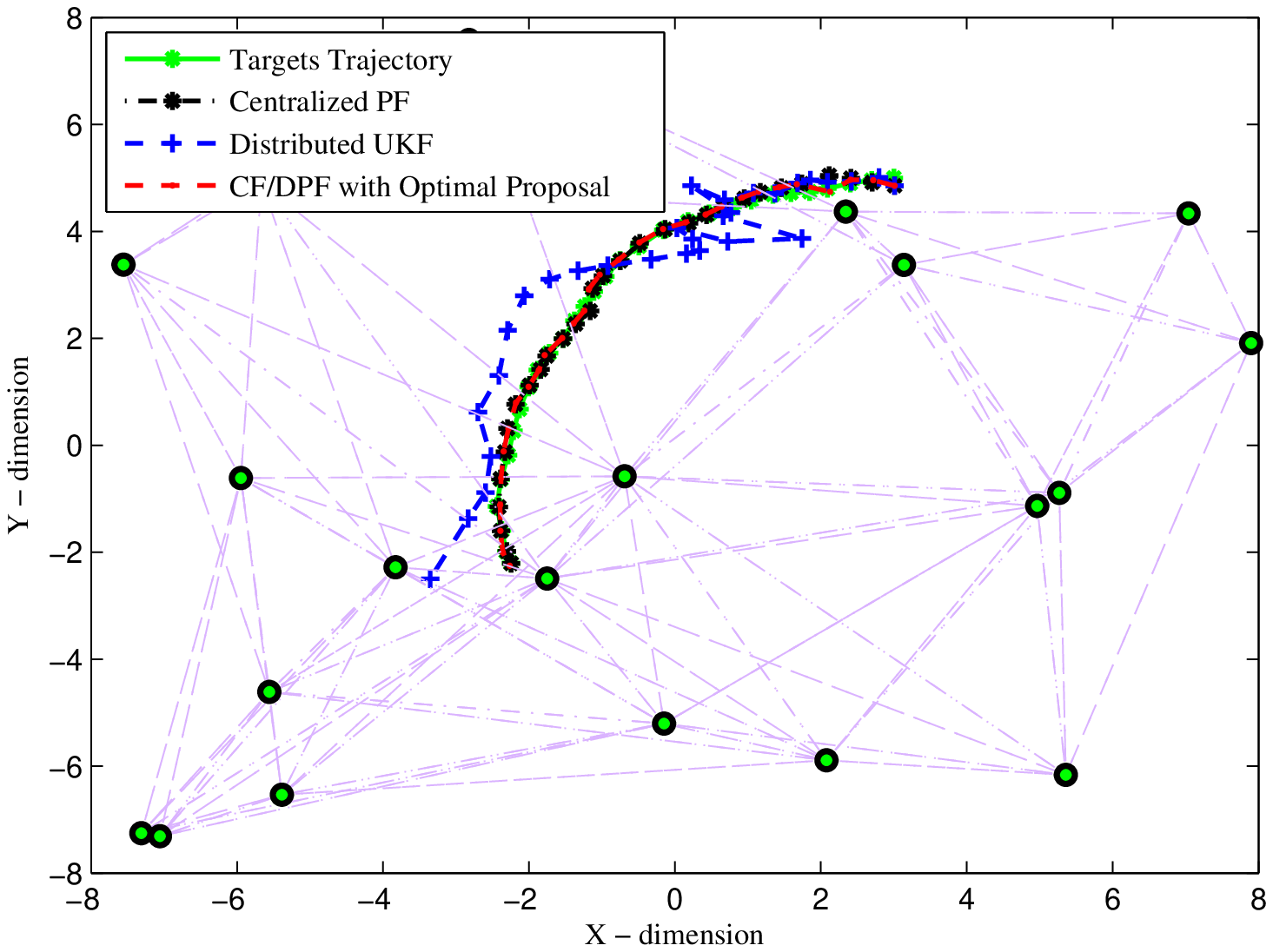}}\quad
\subfigure[]{ \includegraphics[scale=0.5]{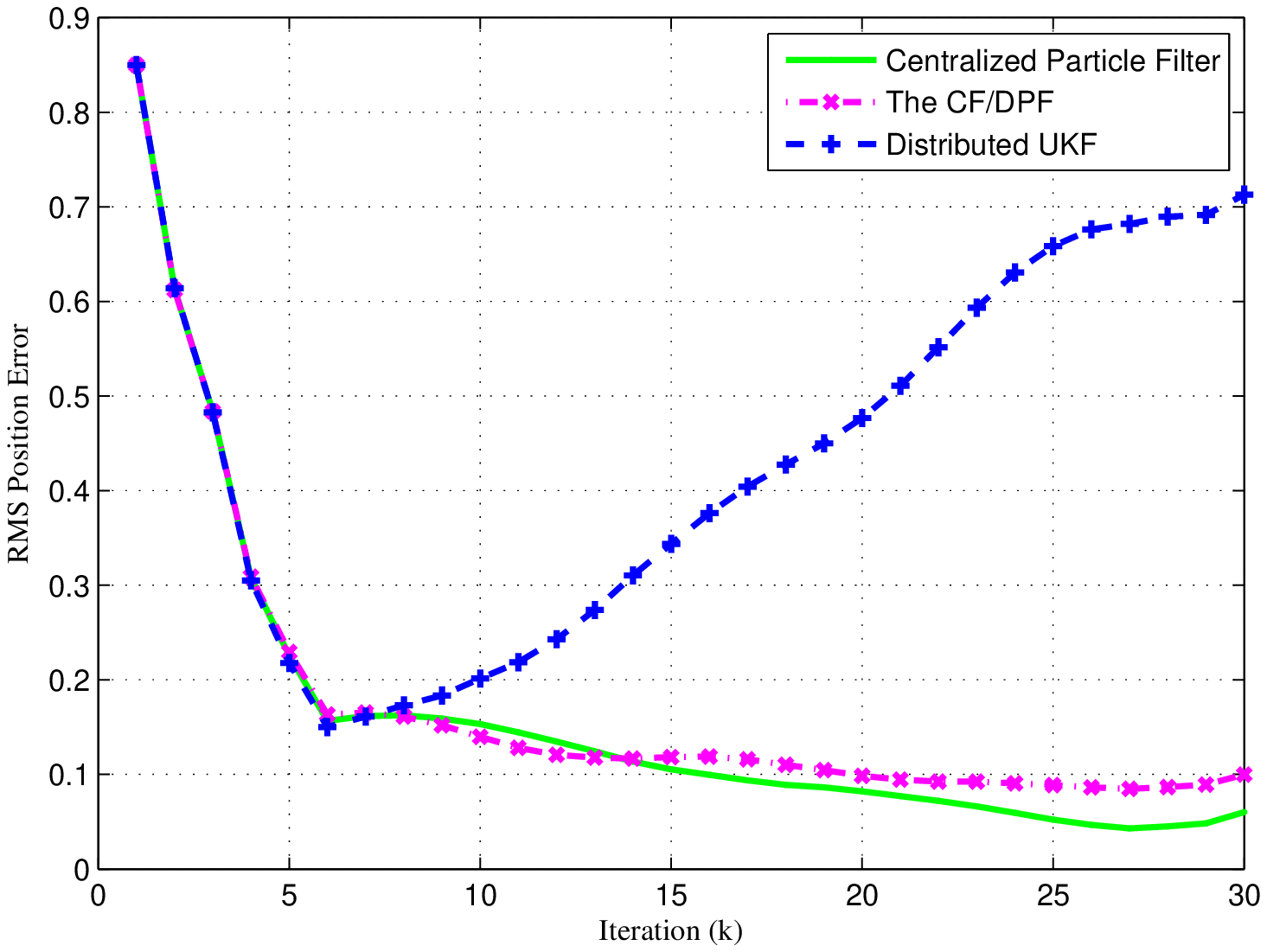}}}
\caption{(a) Robot trajectories estimated from the CF/DPF,
  centralized, and distributed UKF implementations. (b) RMS error plots
  for the three implementations. \label{robot}}
\end{figure}
%
\section{Conclusion} \label{sec:conclusion}
In this paper, we propose a multi-rate framework referred to as the CF/DPF for distributed implementation of the particle filter.
In the proposed framework, two particle filters run at each sensor node. The first filter, referred to as the local filter,
 recursively runs the particle filter based only on the local observations.
We introduce a second particle filter at each node, referred to as the fusion filter, which consistently assimilate local estimates into a global estimate by extracting new information.
Our CF/DPF implementation allows the fusion filter to run at a rate different from that of  the local filters. Achieving consensus between two successive iterations of the localized particle filter is no longer a requirement. The fusion filter and its consensus-step are now separated from the local filters, which enables the consensus step to converge without any time limitations.
Another contribution of the paper is the derivation of the optimal posterior Cram\'er-Rao lower bound (PCRLB) for the distributed architecture based on a recursive procedure involving the local Fisher information matrices (FIM) of the distributed estimators.
Numerical simulations verify the near-optimal performance of the CF/DPF.
The CF/DPF estimates follows the centralized particle filter closely  approaching the PCRLB at the signal to noise ratios that we tested.
\appendices
\section{} \label{app:A}
\begin{proof}[Proof of Theorem~\ref{fprotocol}~\cite{Chong:1990}]
Theorem~\ref{fprotocol} is obtained using: (i) The Markovian property of the state variables; (ii) Assuming that the local observations made at two sensor nodes conditioned on the state variables are independent of each other, and; (iii) Using the Bayes' rule.
Applying the Bayes' rule to Eq.~(\ref{chap3:sec2:eq8}), the posterior distribution can be represented as follows
\begin{eqnarray} \label{eq:ext.1.}
P\big(\x(0\!:\!k)| \z(1\!:\!k) \big) \!\propto\! P\big(\z(k)| \x(k)\big)
 P\big(\x(0\!:\!k) |\z(1\!:\!k\!-\!1)\big).
\end{eqnarray}
Now, using the Markovian property of the state variables, Eq.~(\ref{eq:ext.1.}) becomes
\begin{eqnarray} \label{eq:ext.2.}
P\big(\x(0\!:\!k)| \z(1\!:\!k) \big) \!\propto\!  P\big(\z(k)| \x(k)\big) \times P\big(\x(k) |\x(k\!-\!1)\big) P\big(\x(0\!:\!k\!-\!1) |\z(1\!:\!k\!-\!1)\big).
\end{eqnarray}
Assuming that the local observations made at two sensor nodes conditioned on the state variables are independent of each other Eq.~(\ref{eq:ext.2.}) becomes
\begin{eqnarray} \label{eq:ext.3.}
P\big(\x(0\!:\!k)| \z(1\!:\!k) \big) \!\propto\! \left( \prod_{l=1}^{N} P\big(\z^{(l)}(k)| \x(k)\big) \right)\times P\big(\x(k) |\x(k\!-\!1)\big) P\big(\x(0\!:\!k\!-\!1) |\z(1\!:\!k\!-\!1)\big).
\end{eqnarray}
Using the Bays' rule, the local likelihood function $P\left(\z^{(l)}(k)| \x(k)\right)$ at node $l$, for ($1 \leq l \leq N$)~is 
\begin{eqnarray} \label{eq:ext.4.}
P\left(\z^{(l)}(k)| \x(k)\right) = \frac{P\left(\x(k)|\z^{(l)}(1\!:\!k)\right)}
{P\left(\x(k)|\z^{(l)}(1\!:\!k\!-\!1)\right)} P\left(\z^{(l)}(k)|\z^{(l)}(1\!:\!k\!-\!1)\right).
\end{eqnarray}
Finally, the result (Eq.~(\ref{chap3:sec2:eq8})) is provided by substituting Eq.~(\ref{eq:ext.4.}) in Eq.~(\ref{eq:ext.3.}).
\end{proof}
\section{} \label{app:D}
\begin{proof}[Proof of Theorem~\ref{fprotocol4}]
Following the approach in the proof of Theorem~\ref{fprotocol} (Appendix~\ref{app:A}), we first write the posterior density at iteration $k\!+\!m$ as
\begin{eqnarray} \label{aapd:eq1}
\!\!P\left(\x(0\!:\!k\!+\!m)|\z(1\!:\!k\!+\!m)\right) \propto \frac{\prod_{l=1}^N P
\left( \x(k\!+\!m)|\z^{(l)}(1\!:\!k\!+\!m)\right)}{\prod_{l=1}^N P
\left( \x(k\!+\!m)|\z^{(l)}(1\!:\!k\!+\!m\!-\!1)\right)}
P\left(\x(0\!:\!k\!+\!m)|\z(1\!:\!k\!+\!m\!-\!1)\right)\!.
\end{eqnarray}
Then the last term is factorized as follows
\begin{eqnarray} \label{aapd:eq2}
P\left(\x(0\!:\!k\!+\!m)|\z(1\!:\!k\!+\!m\!-\!1)\right) = P\left(\x(k\!+\!m)|\x(k\!+\!m\!-\!1)\right) P\left(\x(0\!:\!k\!+\!m\!-\!1)|\z(1\!:\!k\!+\!m\!-\!1)\right).
\end{eqnarray}
As in Eq.~(\ref{aapd:eq1}), we continue to expand $ P(\x(0\!:\!k\!+\!m\!-\!1)|\z(1\!:\!k\!+\!m\!-\!1))$ (i.e., the posterior distribution at iteration $k\!+\!m\!-\!1$) all the way back to iteration $k\!+\!1$ to prove Eq.~(\ref{ar-new1}).
\end{proof}
\section{} \label{app:C}
\begin{proof}[Proof of Theorem~\ref{pro3}]
Based on Eq.~(\ref{chap3:sec2:eq8}), term $\log(P(\x(0\!:\!k\!+\!1)|\z(1\!:\!k\!+\!1)))$ is given by
\begin{eqnarray} \label{chap3:sec2:eq88}
&& \log P\big(\x(0\!\!:\!k\!\!+\!\!1)|\z(1\!\!:\!k\!\!+\!\!1)\big) =
\sum_{l=1}^N \log\left(P( \x(k\!\!+\!1)|\z^{(l)}(1\!\!:\!k\!\!+\!1))\right) \nonumber  - \sum_{l=1}^N \log\left(P(\x(k\!+\!1)|\z^{(l)}(1\!:\!k))\right)  \\
&&~~~+  \log\big(P\left(\x(k\!+\!1)|\x(k)\right)\big) \nonumber
+ \log\big(P\left(\x(0\!:\!k)|\z(1\!:\!k)\right)\big).
\end{eqnarray}
Expanding Eq.~(\ref{cpcrlb}) for $\bm{J}(\x(0\!:\!k\!+\!1))$, we get
\begin{eqnarray} \label{ext.crlb.11}
\bm{J}\big(\x(0\!:\!k\!+\!1)\big) =  \mathbb{E} \Bigg\{ -1 \times  \left[
\begin{array}{lll}
\!\!\!\!\Delta^{\x(0:k-1)}_{\x(0:k-1)}  \!\!&\!\! \Delta^{\x(k)}_{\x(0:k-1)} \!\!&\!\! \Delta^{\x(k+1)}_{\x(0:k-1)}\!\!\!\! \\
\!\!\!\!\Delta^{\x(0:k-1)}_{\x(k)}  \!\!&\!\! \Delta^{\x(k)}_{\x(k)} \!\!&\!\! \Delta^{\x(k+1)}_{\x(k)}\!\!\!\! \\
\!\!\!\!\Delta^{\x(0:k-1)}_{\x(k+1)}  \!\!&\!\! \Delta^{\x(k)}_{\x(k+1)} \!\!&\!\! \Delta^{\x(k+1)}_{\x(k+1)}\!\!\!\!
\end{array} \right]
\!\! \log P\big(\x(0\!:\!k\!\!+\!\!1)|\z(1\!:\!k\!+\!1)\big) \Bigg\}
\end{eqnarray}
Substituting Eq.~(\ref{chap3:sec2:eq88}) in Eq.~(\ref{ext.crlb.11}),
it can be shown that
\begin{eqnarray} \label{ext.crlb.11-2}
&&\!\!\!\!\!\!\!\!\!\!\!\!\!\!\!\!\!\!\!\bm{J}\big(\x(0\!:\!k\!+\!1)\big) \!\!=\!\!   \left[
\begin{array}{ccc}
\bm{A}^{11}(k) & \bm{A}^{12}(k) & \bm{0} \\
\bm{A}^{21}(k) & \bm{A}^{22}(k)+ \bm{D}^{11}(k)& \bm{D}^{12}(k) \\
\bm{0} & \bm{D}^{21}(k) & \bm{C}^{22}(k) \\
\end{array} \right],
\end{eqnarray}
where $\bm{A}^{11}(k)$, $\bm{A}^{12}(k)$, $\bm{A}^{21}(k)$ and
$\bm{A}^{22}(k)$ are the same as for Eq.~(\ref{cpcrlb});
$\bm{D}^{11}(k)$, $\bm{D}^{12}(k)$ and $\bm{D}^{21}(k)$ are the same
as in Eqs.~(\ref{ext.crlb.10.1})-(\ref{ext.crlb.10.2}), and;
$\bm{C}^{22}(k)$ is defined in Eq.~(\ref{ext.crlb.14.1}).  Block
$\bm{0}$ stands for a block of all zeros with the appropriate
dimension. The information sub-matrix $\bm{J}\big(\x(k\!+\!1)\big)$
can be calculated as the inverse of the right lower ($n_x\times n_x$)
sub-matrix of $\big[ \bm{J}\big(\x(0\!:\!k\!+\!1)\big) \big]^{-1}$ as
follows
\begin{eqnarray} \label{ext.crlb.12}
\bm{J}\big(\x(k\!+\!1)\big) &=& \bm{C}^{22}(k)-\big[\bm{0} \quad \bm{D}^{21}(k)\big]\times  \left[
\begin{array}{cc}
\bm{A}^{11}(k) & \bm{A}^{12}(k) \\
\bm{A}^{21}(k) & \bm{A}^{22}(k)+\bm{C}^{11}(k) \\
\end{array}
\right]^{-1}
\left[
\begin{array}{c}
\bm{0} \\
\bm{D}^{12}(k) \\
\end{array} \right] \\
& =& \bm{C}^{22}(k)-\bm{D}^{21}(k)\big( \bm{J}\big(\x(k)\big)+\bm{D}^{11}(k)\big)^{-1}\bm{D}^{12}(k), \nonumber
\end{eqnarray}
obtained from the definition of $\bm{J}\big(\x(k)\big)$ based on
Eq.~(\ref{ext.crlb.8}).
\end{proof}
\section{} \label{app:B}
\begin{proof}[Proof of Proposition~\ref{pro2}]
The proof of Proposition~\ref{pro2} is similar to that of
Theorem~\ref{pro3} using the Markovian property of the state variables
and based on the factorization of the joint prediction distribution
$~~~~~~~~~~~~~~~~~~~~~~~~~P\big(\x(0\!:\!k\!+\!1)|\z(1\!:\!k)\big) = P\big(\x(k\!+\!1)|\x(k)\big) P\big(\x(0\!:\!k)|\z(1\!:\!k)\big).$
\end{proof}
\bibliographystyle{IEEEbib}

\end{document}